\documentclass[11pt,onecolumn]{IEEEtran}
\usepackage{amsmath,amssymb,amsthm}
\usepackage[hidelinks]{hyperref}

\usepackage{tikz}
\usetikzlibrary{arrows.meta, positioning, fit, decorations.pathreplacing}

\theoremstyle{definition}
\newtheorem{definition}{Definition}
\newtheorem{assumption}{Assumption}
\theoremstyle{plain}
\newtheorem{lemma}{Lemma}
\newtheorem{proposition}{Proposition}
\newtheorem{theorem}{Theorem}

\title{Homomorphic Encryption with Coset Code Construction}
\author{\IEEEauthorblockN{Haodong Yang}
\IEEEauthorblockA{\textit{EECS}  \textit{Syracuse University}\\
Syracuse, U.S. \\
hyang85@syr.edu}
\and
\IEEEauthorblockN{Venkata Gandikota}
\IEEEauthorblockA{\textit{EECS}  \textit{Syracuse University}\\
Syracuse, U.S. \\
vsgandik@syr.edu}
}

\begin{document}

\maketitle

\begin{abstract}
    Homomorphic encryption enables computation on input data without decrypting it. This paper extends a construction based on scalar message codes, utilizing coset codes to pack $k$ message symbols into a ciphertext vector of length $n$. For a given field and vector length, this packing method reduces the number of ciphertext vectors required to transmit a message record. A compatible family of evaluation codes supports public addition and multiplication operations up to a predetermined degree. We prove that the scheme achieves bounded chosen-plaintext security under the coordinate-label distinguishability assumption.
\end{abstract}

\section{Introduction}

Homomorphic encryption allows a computing party to perform supported operations on encrypted data without possessing the decryption key~\cite{FHEGuide2015}. A data owner can encrypt a set of inputs, send the ciphertexts to a server, and later decrypt the result returned by that server. The goal of computing on encrypted data was articulated by Rivest, Adleman, and Dertouzos through privacy homomorphisms~\cite{RAD1978}. Gentry's construction subsequently established fully homomorphic encryption, supporting general computations on encrypted inputs~\cite{GentryThesis2009}. This capability lets the owner use external computing resources while keeping the inputs and intermediate values confidential from the computing party.

Even a restricted set of homomorphic operations can be useful. Additively homomorphic encryption, such as Paillier's scheme, allows an encrypted sum to be formed from encryptions of individual values~\cite{Paillier1999}. Castelluccia, Mykletun, and Tsudik apply additive homomorphism to the aggregation of encrypted sensor measurements~\cite{CMT2005}. Supporting both addition and multiplication extends the range of possible computations beyond aggregation, including the evaluation of arithmetic expressions~\cite{GentryThesis2009} and encrypted database lookup~\cite{SmartVercauteren2014}. These applications motivate studying how much information a ciphertext can carry and which operations remain possible on that representation.

Coding theory also provides foundations for encryption. McEliece introduced a public-key cryptosystem based on algebraic codes, in which secret code structure enables efficient correction of the errors introduced during encryption~\cite{McEliece1978}. Homomorphic encryption from codes additionally requires operations on ciphertexts to preserve the intended message arithmetic~\cite{Grimaldi2025}. Armknecht et al.~\cite{AAPS11} pursue this goal using noisy evaluation codewords that share a secret set of error-free coordinates. Bogdanov and Lee~\cite{Andrej11} propose another construction using noisy codeword encodings. These works connect encrypted computation with the algebraic properties and recovery procedures of codes.

From an algebraic perspective, homomorphic correctness requires that a supported operation on encrypted representations yield the corresponding operation on the messages after decryption~\cite{RAD1978,FHEGuide2015}. Linear combinations of codewords remain in a linear code, whereas coordinatewise products can lie in a larger space; products and powers of linear codes formalize this behavior~\cite{Randriambololona2015}. In evaluation-based polynomial codes, multiplication of two codewords raises the degree of the encoding polynomials and therefore changes the recovery requirements~\cite{AAPS11}. In lattice-based constructions, arithmetic operations instead increase ciphertext noise, requiring mechanisms to control its growth~\cite{GentryThesis2009,BGV2014}. For example, Brakerski, Gentry, and Vaikuntanathan give leveled constructions whose parameters accommodate a prescribed circuit depth~\cite{BGV2014}. These examples show why the representation and the supported computation range must be designed together.

The security of the algebraic structure has also been examined from the attacker's perspective. Gauthier, Otmani, and Tillich exploit dimensions of coordinatewise product codes to attack the construction of Bogdanov and Lee~\cite{GOT12}. The bounded-use analysis of Armknecht et al.~\cite{AAPS11} separately highlights information exposed by repeated encryptions under one key. These results motivate examining both fresh ciphertexts and the algebraic relations visible through their public evaluation.

The number of message symbols carried by one ciphertext is another important design choice. Packing is already used in homomorphic encryption: Smart and Vercauteren develop single-instruction, multiple-data operations that process several encrypted data elements in parallel~\cite{SmartVercauteren2014}. Gentry, Halevi, and Smart use packed plaintext vectors and routing operations to amortize the cost of homomorphic circuit evaluation~\cite{GHS2012}. In the coding-theoretic setting considered here, we study packing through message cosets, drawing on the coset-coding framework~\cite{coset-code} and the homomorphic error-correcting codes of Yang, Gandikota, and Jones~\cite{YGJHomECC}. Our focus is the extension of the scalar construction of Armknecht et al.~\cite{AAPS11} to vector messages, together with the recovery requirements and security restrictions of that extension.

In the model considered here, an evaluator adds and multiplies ciphertext vectors coordinatewise without the secret key. The receiver decrypts the resulting ciphertext to obtain the corresponding operation on the message vectors. Multiplication can move the clean component into a larger code space, so the construction uses a compatible sequence of codes and a bound on the formal degree of supported computations. Public evaluation produces derived ciphertexts without fresh encryption randomness. Our security analysis considers an attacker who sees a bounded number of fresh encryptions under one key and may perform these public computations. Following the indistinguishability approach to encryption security~\cite[Sec.~5.1]{GM1984}, we ask whether the observer can distinguish encryptions of different messages despite the algebraic relations exposed by evaluation.

\subsection{Our Contributions}
\label{subsec:contributions}

The construction proposed by Armknecht et al.~\cite{AAPS11} encrypts a single message symbol (represented by a field element) within each ciphertext. 
% the authors leave the extension to longer messages as a topic for future research. 
Furthermore, they note a limit on the number of encryption operations that can be performed using the same key. Therefore, if the number of samples is sufficiently large, information leakage may occur, thereby compromising confidentiality.

Drawing upon the homomorphic error-correcting code scheme proposed by Yang, Gandikota, and Jones~\cite{YGJHomECC}, we incorporate a coset code~\cite{coset-code} construction into this encryption framework, thereby extending the method of message representation. Its primary advantage is that each ciphertext can carry multiple message symbols simultaneously; this allows longer messages to be encrypted as a single unit and fully recovered by the recipient. 

For example, consider a record containing three sensor readings, each represented by one symbol over the same finite field. The earlier scalar construction requires three separate encryptions, producing one ciphertext for each symbol. Our construction can encrypt all three symbols together in one ciphertext, provided the chosen coset code supports three message symbols. Thus, at a common field and ciphertext-vector length, the packed representation requires storing and transmitting one evaluation vector instead of three. This approach amortizes the codeword storage overhead across multiple message symbols. A comparison at matched security levels additionally requires code parameters that provide comparable security.

We establish conditions for recovering the entire packed message from the good positions and for preserving message addition and multiplication under public computation. The receiver solves one code-space recovery problem for all message symbols, although its cost still depends on the number of symbols recovered. Both addition and multiplication operate directly on ciphertext vectors. A compatible sequence of code spaces accounts for the growth caused by multiplication, while retaining the same message space and ciphertext-vector length.

Packing reduces the number of fresh ciphertexts required to transmit a fixed number of message symbols, but this representation benefit does not by itself establish stronger security. We reduce bounded chosen-plaintext security to a coordinate-label distinguishing assumption for the exact code and error distributions. We also give a rank test that uses products of zero encryptions to rule out that assumption for some parameter choices. The sample budget, message dimension, and supported computation degree must therefore be assessed together.

\subsection{Organization of the Paper}
\label{subsec:organization}

We organize the paper so that readers can first understand the construction and its assumptions, then see the results and comparisons together, and finally examine their proofs. The starting point is the extension of the single-symbol construction of Armknecht et al.~\cite{AAPS11} to packed vector messages. The definitions and algorithms specify what is being constructed, while the security model separates correctness from confidentiality. We then collect the main guarantees, their limitations, and the concrete code-family parameters in one results section. A separate proofs section supplies the reductions and explicit constructions underlying those claims.

Section~\ref{sec:related-work} reviews earlier work, and Section~\ref{sec:scheme-overview} explains the main idea. Section~\ref{sec:preliminaries} introduces the notation and algebraic requirements, followed by the algorithms in Section~\ref{sec:encryption}. Section~\ref{sec:security-model} defines the security experiments and hardness assumption. Section~\ref{sec:main-results} presents the correctness and security results, the Reed--Muller and Reed--Solomon existence theorems, and the packing and storage comparisons. Section~\ref{sec:proofs} proves these results and gives the detailed constructions and examples. Section~\ref{sec:conclusion} concludes the paper.

\section{Related Work}
\label{sec:related-work}

\subsection{Homomorphic Encryption from Codes}
\label{subsec:related-code-encryption}

Grimaldi~\cite{Grimaldi2025} surveys homomorphic encryption schemes based on coding theory and polynomials. We compare the constructions most closely related to our message representation and public arithmetic.

Homomorphic encryption from codes uses the algebraic structure of codewords to support computations on encrypted messages. Armknecht et al.~\cite{AAPS11} construct a symmetric scheme from evaluation codes. To encrypt a field element $m$, the sender randomly chooses an encoding polynomial $f$ satisfying $f(y)=m$ at a designated message point $y$. The ciphertext contains noisy versions of the evaluations $f(x_1),\ldots,f(x_n)$ at other points. A hidden set of error-free coordinates enables the receiver to recover the message value. Coordinate-wise addition and multiplication implement the corresponding message operations, with multiplication restricted by a prescribed polynomial-degree bound. Their construction encrypts one field element per ciphertext and limits the number of fresh encryptions under one key. Their framework applies to several suitable evaluation-code families, with Reed--Muller codes serving as a concrete instantiation~\cite{AAPS11}.

The choice of a single message symbol is deliberate. In their discussion of special evaluation codes, the authors explain that it simplifies parts of the security analysis and leaves many possible encoding polynomials for the same message, allowing encryption to be randomized~\cite[Sec.~3]{AAPS11}. They explicitly note that larger plaintext spaces are possible and leave that extension for future work. This motivates our comparison with a vector-message coset representation, in which the code space is divided between carrying message symbols and providing random representatives. The amount of randomization retained must be considered when choosing how many symbols to pack together.

Bogdanov and Lee~\cite{Andrej11} proposed a unique approach based on affine code encoding utilizing independent random noise. Their design links arithmetic operations on messages to operations on ciphertext vectors and develops techniques to control decryption errors during homomorphic evaluation. Subsequent cryptanalysis has highlighted the significance of this specific encoding structure. Gauthier, Otmani, and Tillich~\cite{GOT12} proposed an attack against the modified Reed-Solomon codes used in this construction, exploiting dimensional properties of the space generated by coordinate-wise products. More recently, Lage and Kurmann~\cite{LageKurmann2026} reported attacks against Reed-Muller constructions—including specific instances proposed by Armknecht et al.—that enable partial key recovery and, in many cases, message recovery as well.

% These works motivate analyzing the exact encoding and error distribution used by a scheme. 
Challa and Gunta~\cite{ChallaGunta2021Modified} randomize Reed--Muller message encoding by prepending a random string. For a linear encoder, this can be viewed as coset sampling: writing the encoded word as $\boldsymbol rG_R+\boldsymbol mG_M$, fixing $\boldsymbol m$ and varying $\boldsymbol r$ produces the coset $\boldsymbol mG_M+\operatorname{rowspan}(G_R)$. Here row-vector notation separates the random and message coordinates of their encoder. Their expanded-code construction~\cite{ChallaGunta2021Expanded} instead represents a $k$-symbol message by a $k\times n$ ciphertext matrix with secret error positions and a permutation, supporting coordinatewise addition and multiplication. Thus vector messages and randomized codeword representations already occur in earlier code-based proposals.

Polynomial evaluation also underlies the schemes of Dowerah and Krishnaswamy. Their scalar construction samples masking polynomials from a secret ideal and adds discrete Gaussian noise~\cite{DowerahKrishnaswamy2019}. Their later construction encrypts multiple plaintext bits in one ciphertext and supports multiplication through a public tensor evaluation key~\cite{DowerahKrishnaswamy2022}. These schemes recover messages using linear transformations and rounding subject to noise bounds. Under the compatibility conditions below, our construction uses exact recovery of a message coset from secret good coordinates, with errors confined to their complement and public multiplication performed directly on ciphertext coordinates.

Aguilar-Melchor, Dyseryn, and Gaborit~\cite{ADG2025} construct a secret-key scheme from random rank-metric ideal codes, supporting additions, plaintext multiplications, and a prescribed number of ciphertext multiplications. Their security analysis restricts the number of independent ciphertexts published under one key. Their rank-metric construction differs from the fixed secret error locations used here, but it demonstrates that bounded homomorphic computation and bounded ciphertext exposure also arise in other code-based frameworks.

Our construction adapts the synchronized-error framework of~\cite{AAPS11} to vector messages represented by cosets. Its public multiplication procedure uses compatible code spaces, and its confidentiality analysis concerns a bounded number of fresh ciphertexts under an explicitly stated assumption. The framework is formulated for nested linear codes and consistent vector-message maps, independently of the RM and RS examples. Correctness applies to any efficiently sampled family satisfying Definition~\ref{def:compatible-family}; security additionally requires Assumption~\ref{ass:security-iid-dscp} for the exact setup distribution and exposed information. This scope does not include arbitrary linear codes without the required multiplication and recovery properties.

\subsection{Coset Codes and Homomorphic Error-Correcting Codes}
\label{subsec:related-coset-homecc}

Coset codes organize codewords into classes whose members represent the same information. Forney~\cite{coset-code} develops the classical framework for coset coding and its geometric structure. For the present work, we use the following finite-field quotient formulation, also used in~\cite{YGJHomECC}.

\begin{definition}[Cosets and coset codes]
\label{def:related-coset-code}
Let $\mathbb F_q$ be a finite field, and let
$K\subseteq\mathcal C\subseteq\mathbb F_q^n$ be linear codes.
For $\boldsymbol c\in\mathcal C$, its coset modulo $K$ is
\[
    \boldsymbol c+K
       =\{\boldsymbol c+\boldsymbol z:\boldsymbol z\in K\}.
\]
The coset code is the quotient space
\[
    \mathcal C/K
       =\{\boldsymbol c+K:\boldsymbol c\in\mathcal C\}.
\]
Two representatives specify the same coset exactly when their difference belongs to $K$.
\end{definition}

Equivalently, we could choose a complementary subspace $\mathcal C_1\subseteq\mathcal C_2$ such that $\mathcal C_2=\mathcal C_1\oplus K$. Throughout this paper, $\oplus$ denotes an \emph{internal direct sum} of $\mathbb F_q$-vector spaces. Here both summands are subspaces of the same space $\mathcal C_2\subseteq\mathbb F_q^n$, and the direct-sum condition means
\[
    \mathcal C_2=\mathcal C_1+K,
    \qquad \mathcal C_1\cap K=\{\boldsymbol 0\}.
\]
Every coset then has a unique representative in $\mathcal C_1$: each codeword decomposes uniquely into a component in $\mathcal C_1$ and a component in $K$. With a fixed ordered basis of $\mathcal C_1$, this representative specifies $k$ message symbols, where $k=\dim\mathcal C_1=\dim\mathcal C_2-\dim K$. Adding an element of $K$ changes the representative while preserving the represented message.

Yang, Gandikota, and Jones~\cite{YGJHomECC} study homomorphic error-correcting codes, which combine recovery from errors with compatibility between encoded operations and message arithmetic. For a supported operation, computing on encoded inputs must produce an encoding of the result of that operation on the messages. Their manuscript develops coset constructions using Reed--Muller, Reed--Solomon, and Hermitian codes to support message addition and coordinate-wise multiplication while retaining error-correction
properties.

Their use of canonical representatives provides the basis for our message representation. Public multiplication on randomly masked codewords additionally requires that products of arbitrary representatives preserve the intended message. We formulate this requirement through compatible message maps on a sequence of code spaces and verify it for the two constructions below. While error correction involves recovering a message from a corrupted representation, encryption also requires hiding the message from observers; consequently, the recovery properties of coset codes must be complemented by an independent security analysis.

\section{The Idea of the Encryption Scheme}
\label{sec:scheme-overview}

This section outlines the core idea of the proposed scheme. The fundamental concept involves using a coset to jointly represent multiple message symbols, where distinct codewords within the coset correspond to the same message vector. During encryption, a codeword is randomly selected, and errors are introduced outside a predefined set of ``good positions.'' The receiver utilizes these specific positions to identify the coset to which the message belongs and subsequently recovers the complete message vector.

We first introduce the subspaces that carry the message and the random mask. Next, we describe how the encryption process combines these components with errors and how the receiver recovers the message. Finally, we present the public addition operation and the conditions required to perform multiplication without decrypting the input data. Figure~\ref{fig:scheme-overview} summarizes the transformation from plaintext to ciphertext, the public computation process, and the receiver's recovery process. Relevant formal definitions and compatibility conditions are provided in Section~\ref{sec:preliminaries}.

\subsection{Message and Mask Subspaces}
\label{subsec:overview-key}

Let $\mathcal C_2 = \mathcal C_1 \oplus K \subseteq \mathbb F_q^n$ be an internal direct sum. The subspace $\mathcal C_1$ contains message representatives, while $K$ contains random masks. Each ``clean'' codeword can be uniquely decomposed into these two components.

A fixed linear bijection maps each message vector $\boldsymbol m \in \mathbb F_q^k$ to a representative $\boldsymbol c_{\boldsymbol m} \in \mathcal C_1$ (where $k = \dim \mathcal C_1$), with the corresponding message coset being $\boldsymbol c_{\boldsymbol m} + K$. Within this coset, adding any element from $K$ alters the codeword but leaves the represented message unchanged. Consequently, distinct codewords chosen at random can carry the same underlying message vector.

The sender and receiver share the message basis, bases specifying the mask spaces at all supported computation levels, and a set $I$ of good positions where no errors are introduced; the remaining positions constitute $B$. The nested mask bases can be stored as prefixes of a single basis, as specified in the setup algorithm. Both coordinate sets remain fixed under a key. Field representation, dimensions, algorithms, and the supported computation range are public information; this is a symmetric encryption scheme. The receiver derives the appropriate decoder from the secret key and the public computation level.

\subsection{Encryption and Message Recovery}
\label{subsec:overview-recovery}

For encryption, the message is first represented as $\boldsymbol c_{\boldsymbol m}$. A uniformly random mask $\boldsymbol z \in K$ is chosen, and an error vector $\boldsymbol e$ is independently selected; this error vector is distributed uniformly and independently over $B$ and is zero over $I$. The resulting ciphertext is denoted as $\boldsymbol w_{\boldsymbol m}$ and transmitted.
\begin{equation}
    \boldsymbol w_{\boldsymbol m}
       =\underbrace{\boldsymbol c_{\boldsymbol m}}_{\text{message representative}}
        +\underbrace{\boldsymbol z}_{\text{random mask}}
        +\underbrace{\boldsymbol e}_{\text{error}}.
    \label{eq:overview-encryption}
\end{equation}
The first two terms, $\boldsymbol c_{\boldsymbol m} + \boldsymbol z$, constitute a random codeword within the message coset. Errors may cause the ciphertext to fall outside $\mathcal C_2$. Each new encryption operation employs fresh, mutually independent mask and error values. Consequently, while $\boldsymbol c_{\boldsymbol m}$ is determined by the message, $\boldsymbol w_{\boldsymbol m}$ depends on the randomness inherent in the encryption process.

The receiver reads the information at the ``good'' positions and treats the other positions as erasures. The choice of the secret key ensures that if any two ``clean'' codewords match at these good positions, they belong to the same message coset. Consequently, these positions uniquely identify a coset, even if they do not uniquely identify a specific clean codeword. The unique representative of this coset in $\mathcal C_1$ determines the original message vector; there is no need to reconstruct the original random mask. If no clean codeword matches the received information at the good positions, the input is rejected.

\begin{figure}[tbp]
    \centering
    \begin{tikzpicture}[
        font=\small,
        data/.style={align=center, inner sep=3pt},
        heading/.style={align=center, text width=2.5cm},
        flow/.style={-{Stealth[length=2mm]}, line width=0.6pt},
        operation/.style={font=\scriptsize, align=center},
        group/.style={decorate, decoration={brace, amplitude=4pt}}
    ]
        \node[heading] at (0,0.75) {Plaintext};
        \node[heading] at (3,0.75) {Message\\representative};
        \node[heading] at (6,0.75) {Randomized\\codeword};
        \node[heading] at (9,0.75) {Ciphertext};
        \node[heading] at (12,0.75) {Plaintext};

        \node[data] (message) at (0,0) {$\boldsymbol m$};
        \node[data] (representative) at (3,0)
            {$\boldsymbol c_{\boldsymbol m}$};
        \node[data] (masked) at (6,0)
            {$\boldsymbol c_{\boldsymbol m}+\boldsymbol z$};
        \node[data] (cipher) at (9,0)
            {$\boldsymbol w_{\boldsymbol m}$};
        \node[data] (recovered) at (12,0) {$\boldsymbol m$};

        \draw[flow] (message) --
            node[operation, above=2pt] {$\mathsf{Encode}$}
            (representative);
        \draw[flow] (representative) --
            node[operation, above=2pt] {$+\boldsymbol z$}
            node[operation, below=2pt] {Add mask}
            (masked);
        \draw[flow] (masked) --
            node[operation, above=2pt] {$+\boldsymbol e$}
            node[operation, below=2pt] {Add errors}
            (cipher);
        \draw[flow] (cipher) --
            node[operation, above=2pt]
                {$\mathsf{Decrypt}_{\mathsf{sk}}$}
            (recovered);

        \draw[group] (-0.45,1.4) -- (8.9,1.4)
            node[midway, above=7pt] {\textbf{Encryption}};
        \draw[group] (9.1,1.4) -- (12.45,1.4)
            node[midway, above=7pt] {\textbf{Decryption}};
    \end{tikzpicture}

    \smallskip
    {\small (a) From plaintext to ciphertext and back.\par}
    \bigskip

    \begin{tikzpicture}[
        font=\small,
        data/.style={align=center, inner sep=3pt},
        flow/.style={-{Stealth[length=2mm]}, line width=0.6pt},
        operation/.style={font=\scriptsize, align=center}
    ]
        \node[data] (inputs) at (0,0)
            {$(\boldsymbol w_{\boldsymbol m},\boldsymbol w_{\boldsymbol m'})$};
        \node[data] (sum) at (5,0)
            {$\boldsymbol w_{\boldsymbol m}+\boldsymbol w_{\boldsymbol m'}$};
        \node[data] (result) at (10.6,0)
            {$\boldsymbol m+\boldsymbol m'$};
        \draw[flow] (inputs) --
            node[operation, above=2pt] {$+$}
            node[operation, below=2pt] {Public addition}
            (sum);
        \draw[flow] (sum) --
            node[operation, above=2pt]
                {$\mathsf{Decrypt}_{\mathsf{sk}}$}
            (result);
    \end{tikzpicture}

    \smallskip
    {\small (b) Public addition, followed by decryption using the key.\par}
    \bigskip

    \begin{tikzpicture}[
        font=\small,
        data/.style={align=center, inner sep=3pt},
        flow/.style={-{Stealth[length=2mm]}, line width=0.6pt},
        operation/.style={font=\scriptsize, align=center}
    ]
        \node[data] (inputs) at (0,0) {$(\boldsymbol w_{\boldsymbol m},\boldsymbol w_{\boldsymbol m'})$};
        \node[data] (product) at (5,0) {$\boldsymbol w_{\times}=\boldsymbol w_{\boldsymbol m}\odot\boldsymbol w_{\boldsymbol m'}$};
        \node[data] (result) at (10.6,0) {$\boldsymbol m\odot\boldsymbol m'$};
        \draw[flow] (inputs) -- node[operation, above=2pt] {$\odot$} node[operation, below=2pt] {Public multiplication} (product);
        \draw[flow] (product) -- node[operation, above=2pt] {$\mathsf{Decrypt}^{\mathrm{prod}}_{\mathsf{sk}}$} (result);
    \end{tikzpicture}

    \smallskip
    {\small (c) Public multiplication, followed by decryption using the product-code decoder.\par}
    \caption{Encoding, encryption, decryption, and public computations. In (a), $\boldsymbol z\in K$ and $\boldsymbol e$ is zero on the good positions $I$. Decryption uses these positions to recover the message coset and its message; it need not recover the original randomized codeword. Panels (b) and (c) perform ciphertext operations before decryption. In (c), $\odot$ denotes coordinatewise multiplication, and correctness requires compatible message cosets and a recoverable product code, as described in Section~\ref{subsec:overview-computation}. Only the receiver uses the key to decrypt the resulting ciphertext.}
    \label{fig:scheme-overview}
\end{figure}
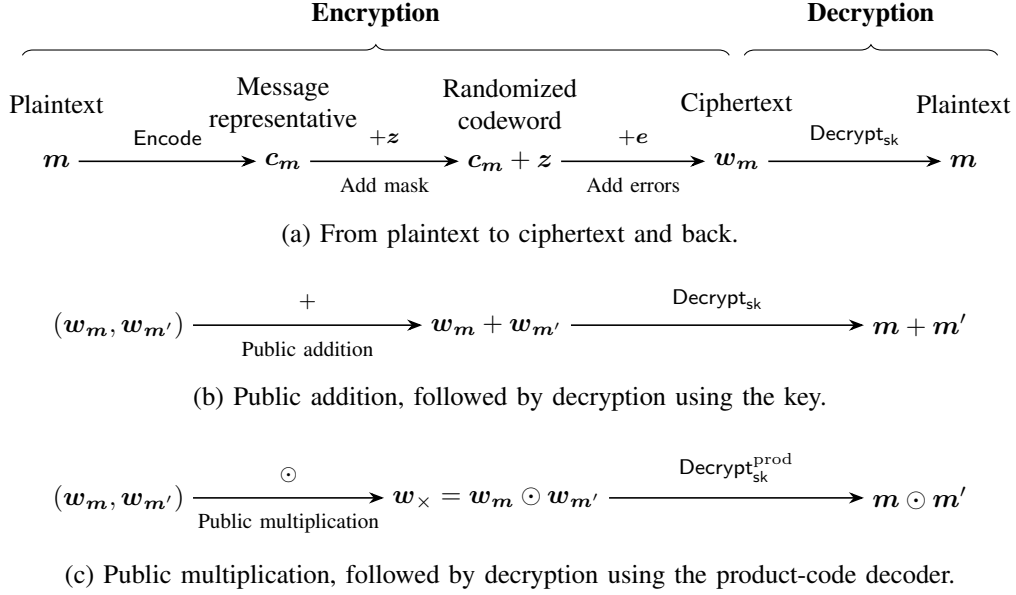

\subsection{Homomorphic Computation}
\label{subsec:overview-computation}

% Addition operations are performed directly on the ciphertext vectors. The values representing the message, the masks, and the errors are added within their respective subspaces, such that the result represents the sum of the messages.

For multiplication, the evaluator computes the coordinatewise product directly:
\begin{equation}
    \boldsymbol w_{\times}=\boldsymbol w_{\boldsymbol m}\odot\boldsymbol w_{\boldsymbol m'}.
    \label{eq:overview-public-multiplication}
\end{equation}
% The evaluator neither decrypts the inputs nor re-encrypts the outputs. 
On the set of positions $I$, the errors in both inputs are zero; consequently, the product restricted to these positions equals the coordinate-wise product of the two ``clean" codewords (error-free codewords) restricted to $I$. It follows that the error introduced by multiplication is confined to the set of positions $B$.

The chosen family of coset codes must also ensure that, for any choice of masks $\boldsymbol z, \boldsymbol z' \in K$, their clean product represents $\boldsymbol m\odot\boldsymbol m'$. Specifically, product terms involving the masks must not alter the output message. In the output code, ``good'' positions (i.e., error-free positions) must be capable of distinguishing between different message cosets. These additional requirements concern compatibility and recovery; the internal direct sum $\mathcal C_2=\mathcal C_1\oplus K$ alone is insufficient to satisfy them. The receiver uses the corresponding output-code decoder only when it needs the computed message.

The size of the product code may exceed that of $\mathcal C_2$. Consequently, repeated multiplication requires compatible code spaces and recovery conditions at each supported level. A public level label tracks the formal degree of the computation and identifies the corresponding output code space for the receiver. The construction specifies a maximum supported degree, rather than a bound on the number of multiplication gates. The code families below verify these compatibility and recovery conditions; an arbitrary coset code need not support multiplication within its original space.

Ciphertexts resulting from public addition and multiplication are derived from existing inputs and do not introduce new encryption randomness. Their outputs need not exhibit the distributional properties of "fresh encryption." While the "fresh sample budget" tracks encryption instances independently randomized under the same key, security analysis must account for information leaked through the public evaluation process. Packing more symbols into a single ciphertext does not eliminate the inherent limitations associated with repeated encryption under the same key.

\section{Notation and Preliminaries}
\label{sec:preliminaries}

This section begins by defining linear codes and coset codes, followed by an introduction to code products, evaluation codes, and encoding. Next, we outline the conditions required for message recovery from synchronization errors and for performing public computations on ciphertexts.

\subsection{Codes, Cosets, and Encoding}
\label{subsec:prelim-vectors}

All vector spaces and linear maps are over the finite field $\mathbb F_q$, where $q$ is a prime power. Vectors are columns, $[n]=\{1,\ldots,n\}$, and $\boldsymbol v_I$ denotes restriction to $I\subseteq[n]$ in increasing coordinate order. Write $\operatorname{supp}(\boldsymbol v)=\{i\in[n]:v_i\ne 0\}$, $\operatorname{wt}(\boldsymbol v)=|\operatorname{supp}(\boldsymbol v)|$, and $d_H(\boldsymbol u,\boldsymbol v)=\operatorname{wt}(\boldsymbol u-\boldsymbol v)$ for support, Hamming weight, and Hamming distance, respectively. For a finite nonempty set $S$, $U(S)$ denotes its uniform distribution and $x\leftarrow U(S)$ denotes sampling from it.

\begin{definition}[Linear code]
\label{def:linear-code}
An $[n,k,d]_q$ linear code $\mathcal C$ is a $k$-dimensional subspace of $\mathbb F_q^n$ with minimum Hamming distance $d(\mathcal C)=d$. A generator matrix $G\in\mathbb F_q^{n\times k}$ has columns forming a basis of $\mathcal C$, so $\mathcal C=\{G\boldsymbol u:\boldsymbol u\in\mathbb F_q^k\}$.
\end{definition}

To represent a message by several possible codewords, we partition a linear code into cosets of a subcode.

\begin{definition}[Cosets and coset codes]
\label{def:coset-code}
Let $K\subseteq\mathcal C_2\subseteq\mathbb F_q^n$ be linear codes, with $h=\dim K$ and $k=\dim\mathcal C_2-h\geq 1$. For $\boldsymbol c\in\mathcal C_2$, its coset modulo $K$ and the resulting coset code are
\[
    \boldsymbol c+K=\{\boldsymbol c+\boldsymbol z:\boldsymbol z\in K\},
    \qquad \mathcal C_2/K=\{\boldsymbol c+K:\boldsymbol c\in\mathcal C_2\}.
\]
Two codewords belong to the same coset exactly when their difference lies in $K$. These cosets partition $\mathcal C_2$: each contains $q^h$ codewords, and there are $q^k$ cosets. The quotient is a $k$-dimensional vector space with operations
\[
    (\boldsymbol c+K)+(\boldsymbol c'+K)=(\boldsymbol c+\boldsymbol c')+K,
    \qquad \alpha(\boldsymbol c+K)=\alpha\boldsymbol c+K.
\]
Its elements are cosets, rather than individual length-$n$ vectors. To select representatives, choose $\mathcal C_1\subseteq\mathcal C_2$ satisfying the internal direct sum
\[
    \mathcal C_2=\mathcal C_1\oplus K,
    \qquad \mathcal C_1\cap K=\{\boldsymbol 0\}.
\]
Here $\oplus$ means that every codeword in $\mathcal C_2$ is uniquely a sum of a vector in $\mathcal C_1$ and a vector in $K$. Consequently, $\dim\mathcal C_1=k$, and every coset has a unique representative in $\mathcal C_1$.
\end{definition}

The dimension $k$ determines how many message symbols a coset carries, while $h$ determines the number of independent field symbols used to select a random representative within it. Thus a length-$n$ codeword from $\mathcal C_2$ contains both message and masking components, with $\dim\mathcal C_2=k+h$.

\begin{definition}[Product of linear codes]
\label{def:code-product}
For equal-length vectors, define $\boldsymbol u\odot\boldsymbol v=(u_i v_i)_i$. The coordinatewise product, or Schur product, of linear codes $\mathcal C,\mathcal C'\subseteq\mathbb F_q^n$ is
\[
    \mathcal C\star\mathcal C'=\operatorname{span}_{\mathbb F_q}\{\boldsymbol c\odot\boldsymbol c':\boldsymbol c\in\mathcal C,\ \boldsymbol c'\in\mathcal C'\}.
\]
Thus $\odot$ denotes a vector product, whereas $\star$ denotes all linear combinations of such products. Unlike sums of codewords in one linear code, products need not remain in that code: $\mathcal C\star\mathcal C$ need not be contained in $\mathcal C$. We also use $\odot$ for coordinatewise multiplication of length-$k$ messages.
\end{definition}

Evaluation codes provide a common setting for these products, including the code families considered in~\cite{AAPS11, YGJHomECC}.

\begin{definition}[Evaluation code]
\label{def:evaluation-code}
Let $\Omega$ be a set, let $X=(x_1,\ldots,x_n)$ be an ordered tuple of distinct points in $\Omega$, and let $\mathcal V$ be a finite-dimensional linear space of functions $\Omega\to\mathbb F_q$. For a function $f$, define its evaluation vector and the associated code by
\begin{align*}
    \operatorname{ev}_X(f)&=(f(x_1),\ldots,f(x_n))^{\mathsf T},\\
    \mathcal C_X(\mathcal V)&=\operatorname{ev}_X(\mathcal V)
       =\{\operatorname{ev}_X(f):f\in\mathcal V\}\subseteq\mathbb F_q^n.
\end{align*}
Evaluation is linear, so $\mathcal C_X(\mathcal V)$ is a linear code, with
\[
    \dim\mathcal C_X(\mathcal V)=\dim\mathcal V-\dim\ker(\operatorname{ev}_X|_{\mathcal V}).
\]
Its dimension equals $\dim\mathcal V$ precisely when evaluation is injective on $\mathcal V$; in general, different functions may give the same codeword.
\end{definition}

For example, evaluating univariate polynomials of degree less than $r$ at $n$ distinct points of $\mathbb F_q$, with $1\leq r\leq n\leq q$, gives a Reed--Solomon code of parameters $[n,r,n-r+1]_q$. More generally, pointwise function multiplication satisfies
\begin{equation}
    \operatorname{ev}_X(fg)=\operatorname{ev}_X(f)\odot\operatorname{ev}_X(g).
    \label{eq:prelim-evaluation-product}
\end{equation}
For function spaces $\mathcal V,\mathcal W$ on the same set, put $\mathcal V\mathcal W=\operatorname{span}_{\mathbb F_q}\{fg:f\in\mathcal V,\ g\in\mathcal W\}$. Linearity of evaluation then gives $\mathcal C_X(\mathcal V)\star\mathcal C_X(\mathcal W)=\mathcal C_X(\mathcal V\mathcal W)$. This describes the product code, but compatibility with the intended message multiplication remains an additional requirement.

\begin{definition}[Encoding and message extraction]
\label{def:message-map}
Fix an ordered basis of $\mathcal C_1$, with generator matrix $G_1\in\mathbb F_q^{n\times k}$. The canonical encoder is the linear bijection
\begin{equation}
    \mathsf{Encode}:\mathbb F_q^k\longrightarrow\mathcal C_1,
    \qquad \mathsf{Encode}(\boldsymbol m)=G_1\boldsymbol m=\boldsymbol c_{\boldsymbol m}.
    \label{eq:prelim-encoding}
\end{equation}
The corresponding coset encoder is $\overline{\mathsf{Encode}}(\boldsymbol m)=\boldsymbol c_{\boldsymbol m}+K\in\mathcal C_2/K$. In the reverse direction, the message-extraction map $\pi:\mathcal C_2\to\mathbb F_q^k$ is
\begin{equation}
    \pi(\boldsymbol c_{\boldsymbol m}+\boldsymbol z)=\boldsymbol m,
    \qquad \boldsymbol z\in K.
    \label{eq:prelim-message-map}
\end{equation}
This map is linear, $\ker\pi=K$, and $\pi\circ\mathsf{Encode}=\operatorname{id}_{\mathbb F_q^k}$. It assigns the same message to every representative of a coset, and its induced map on $\mathcal C_2/K$ is the inverse of the coset encoder.
\end{definition}

The encoding process maps a message to its canonical codeword or message coset, while the message extraction process performs the inverse operation on the "clean" representative. We denote the clean codeword as $\boldsymbol c$ and the ciphertext as $\boldsymbol w$. The randomized encryption process $\mathsf{Encrypt}$ also selects a mask and introduces noise; consequently, the noisy ciphertext does not necessarily fall within the domain of $\pi$, necessitating the use of the recovery conditions described below. The choice of $G_1$ determines the message coordinates, and the compatibility of these coordinates with the multiplication operation must be verified.

\paragraph{Average-free exponents.}
The Reed--Solomon construction uses integer exponent sets that separate pure powers from mixed products. We use the following average-free condition~\cite[Sec.~3.2]{ABHS19AverageFree}.

\begin{definition}[Average-free exponent set]
\label{def:average-free-exponents}
For an integer $D\geq2$, a set $S\subseteq\mathbb Z$ is $D$-average-free if, for every $2\leq j\leq D$ and $a,a_1,\ldots,a_j\in S$,
\begin{equation}
    a_1+\cdots+a_j=ja
       \quad\Longrightarrow\quad a_1=\cdots=a_j=a.
    \label{eq:average-free-condition}
\end{equation}
The summands may repeat, and the equation is over the integers, not over $\mathbb F_q$. For $D=2$, this is equivalent to being 3-AP-free. Define
\[
    a_D(N)=\max\{|S|:S\subseteq[N]\text{ is }D\text{-average-free}\}.
\]
\end{definition}

For a positive integer $j$ and a nonempty integer set $S$, let $jS=\{a_1+\cdots+a_j:a_i\in S\}$ denote its $j$-fold sumset, allowing repeated summands. This notation concerns integer exponents, not addition in the finite field.

\subsection{Recovery and Public Computation}
\label{subsec:prelim-recovery}

Let $I\subseteq[n]$ be the good positions, $T=|I|$, $B=[n]\setminus I$, and $b=n-T$. Define
\begin{equation}
    E_B=\{\boldsymbol e\in\mathbb F_q^n:\boldsymbol e_I=\boldsymbol 0\}
       =\{\boldsymbol e\in\mathbb F_q^n:\operatorname{supp}(\boldsymbol e)\subseteq B\}.
    \label{eq:prelim-error-space}
\end{equation}
The set $I$ stays fixed under one key, and the receiver treats $B$ as erasures. Fresh masks and errors are sampled from $U(K)$ and $U(E_B)$, independently across draws conditional on the fixed setup and key. Sampling from $U(E_B)$ gives independent uniform entries on $B$ and zeros on $I$; entries on $B$ may also be zero.

\begin{lemma}[Recovery of a message coset]
\label{lem:prelim-recovery}
Let $\varphi:\mathcal C\to\mathbb F_q^k$ be a surjective linear map on a linear code $\mathcal C\subseteq\mathbb F_q^n$, with $\mathcal K=\ker\varphi$. For all $\boldsymbol c\in\mathcal C$ and $\boldsymbol e\in E_B$, the message $\varphi(\boldsymbol c)$ is determined by $(\boldsymbol c+\boldsymbol e)_I$ if and only if
\begin{equation}
    \mathcal C\cap E_B\subseteq\mathcal K.
    \label{eq:prelim-recovery-condition}
\end{equation}
\end{lemma}
The proof is given in Section~\ref{subsec:scheme-correctness}.

The criterion defines a linear recovery map $\boldsymbol c_I\mapsto\varphi(\boldsymbol c)$ without requiring recovery of the original codeword. Fresh ciphertexts use $(\mathcal C,\mathcal K)=(\mathcal C_2,K)$. The stronger bound $b<d(\mathcal C)$ suffices to recover the entire codeword.

\begin{definition}[Compatible family of message codes]
\label{def:compatible-family}
For a positive integer $\Delta$, let $\mathcal C^{(\ell)}$ denote the code at computation level $\ell$. These codes are nested, with fresh code $\mathcal C^{(1)}=\mathcal C_2$:
\[
    \mathcal C^{(1)}=\mathcal C_2\subseteq\mathcal C^{(2)}\subseteq\cdots\subseteq\mathcal C^{(\Delta)}\subseteq\mathbb F_q^n.
\]
Equip these codes with surjective linear maps $\pi_\ell:\mathcal C^{(\ell)}\to\mathbb F_q^k$, with $\pi_1=\pi$ and $K_\ell=\ker\pi_\ell$. For indices $\ell,r,s\in[\Delta]$, require
\begin{align}
    \left.\pi_s\right|_{\mathcal C^{(r)}}&=\pi_r &&(r\leq s),\label{eq:prelim-level-compatibility}\\
    \mathcal C^{(r)}\star\mathcal C^{(s)}&\subseteq\mathcal C^{(r+s)} &&(r+s\leq\Delta),\label{eq:prelim-product-inclusion}\\
    \pi_{r+s}(\boldsymbol c\odot\boldsymbol c')&=\pi_r(\boldsymbol c)\odot\pi_s(\boldsymbol c') &&(r+s\leq\Delta),\label{eq:prelim-message-product}\\
    \mathcal C^{(\ell)}\cap E_B&\subseteq K_\ell &&(\ell\in[\Delta]).\label{eq:prelim-level-recovery}
\end{align}
Here $\left.\pi_s\right|_{\mathcal C^{(r)}}$ denotes the restriction of the domain of $\pi_s$ from $\mathcal C^{(s)}$ to its subcode $\mathcal C^{(r)}$, leaving the codomain $\mathbb F_q^k$ unchanged. Explicitly, for $r\leq s$, this restricted map sends $\boldsymbol c\in\mathcal C^{(r)}$ to $\pi_s(\boldsymbol c)$. Equation~\eqref{eq:prelim-level-compatibility} requires $\pi_s(\boldsymbol c)=\pi_r(\boldsymbol c)$ for every such codeword, so regarding it as an element of a larger code space does not change its message. This restricts the allowable input codewords, whereas $\boldsymbol c_I$ restricts a vector to selected coordinates.

The product identity must hold for every $\boldsymbol c\in\mathcal C^{(r)}$ and $\boldsymbol c'\in\mathcal C^{(s)}$, including masked representatives. These conditions ensure consistent message maps, compatible products, and recovery from the same good positions.
\end{definition}

Compatibility gives $K_1=K$ and $\mathcal C^{(\ell)}=\mathcal C_1\oplus K_\ell$, with the same canonical representatives at every level. For $r+s\leq\Delta$, it also gives $K_r\star\mathcal C^{(s)}+\mathcal C^{(r)}\star K_s\subseteq K_{r+s}$, so mask products cannot change the output message. Errors remain in $E_B$ because
\begin{equation}
    (\boldsymbol c+\boldsymbol e)\odot(\boldsymbol c'+\boldsymbol e')
      =\boldsymbol c\odot\boldsymbol c'
       +\underbrace{\boldsymbol c\odot\boldsymbol e'+\boldsymbol e\odot\boldsymbol c'+\boldsymbol e\odot\boldsymbol e'}_{\in E_B},
    \label{eq:prelim-product-errors}
\end{equation}
for $\boldsymbol e,\boldsymbol e'\in E_B$. Together with Lemma~\ref{lem:prelim-recovery}, this justifies public multiplication under the stated conditions; it does not establish security for a concrete code family.

Fresh ciphertexts, including encrypted constants, are at level $1$. For inputs at levels $r$ and $s$, the level becomes $\max\{r,s\}$ after addition and $r+s$ after multiplication, provided $r+s\leq\Delta$. These labels bound the formal degree, rather than the circuit depth; for example, two successive squarings have depth two and level four. The evaluator requires neither secret message maps nor decoders. The parameter $L$ limits the number of independently randomized fresh encryptions exposed under the same key. Public evaluation produces derived ciphertexts without consuming fresh encryption samples, and its outputs need not have the distribution of fresh encryptions. Thus $\Delta$ bounds supported computations, while a secure value of $L$ depends on the code family, setup distribution, and hardness assumption.

\section{The Encryption Scheme}
\label{sec:encryption}

This section defines the scheme outlined in Section~\ref{sec:scheme-overview}, using the compatible code spaces of Definition~\ref{def:compatible-family}. We specify key generation, encryption, decryption, and the two public operations. The correctness and security guarantees are collected in Section~\ref{sec:main-results}, with their proofs in Section~\ref{sec:proofs}. The algorithms assume an efficiently sampled compatible instance; a concrete code family and its setup distribution require separate analysis.

\subsection{Description and Encryption/Decryption Algorithms}
\label{subsec:scheme-algorithms}

Consider a publicly specified family of compatible instances that allows for efficient sampling. Let $\lambda$ denote the target security parameter, $\Delta$ the supported bound on the degree of the forms, and $L$ the proposed new bound on exposure to the encryption. The family specifies how its sampler selects the domain, code space, encoding, and "good" positions for the supported parameters. The algorithms described below assume that the instances provided by this sampler satisfy Definition \ref{def:compatible-family}.

A ciphertext is a tuple $(\boldsymbol w, \ell) \in \mathbb{F}_q^n \times [\Delta]$, where $\ell$ is the public computation level introduced in the preliminaries. The ciphertext represents $\boldsymbol m \in \mathbb{F}_q^k$ if the following condition holds:
\begin{equation}
\boldsymbol w \in \mathsf{Encode}(\boldsymbol m) + K_\ell + E_B.
\label{eq:scheme-valid-ciphertext}
\end{equation}
This supplements the ciphertext vector from the overview with a level label. All inputs to an evaluation operation must use the same key. The encryption and decryption processes implicitly use public parameters. Algorithms receiving a ciphertext reject invalid vector dimensions, field elements, or level labels; $\bot$ denotes rejection. The algorithm comprises five distinct components: $\mathsf{Setup}$ (initialization), $\mathsf{Encrypt}$ (encryption), $\mathsf{Decrypt}$ (decryption), $\mathsf{Add}$ (addition), and $\mathsf{Mult}$ (multiplication), as detailed below.

\medskip
\noindent$\mathsf{Setup}(1^\lambda,\Delta, L)$: Run the selected family sampler to obtain a compatible instance; return $\bot$ if the parameters are unsupported. Let $h_\ell = \dim K_\ell$, such that $h_1 = h$. Compatibility implies $K_1 \subseteq \cdots \subseteq K_\Delta$. Select a generator matrix $G_{K,\Delta}$ for $K_\Delta$ such that, for each $\ell$, its first $h_\ell$ columns (denoted $G_{K,\ell}$) form a basis for $K_\ell$. Such a basis can be obtained by incrementally extending the basis along the sequence of nested spaces. Let $G_1$ be the generator matrix defining the $\mathsf{Encode}$ operation (as shown in \eqref{eq:prelim-encoding}). Output the public parameters and the shared secret key.
    \begin{align}
        \mathsf{pp}&=\bigl(q,n,k,T,(h_\ell)_{\ell=1}^{\Delta},\Delta,L\bigr),\label{eq:scheme-public-parameters}\\
        \mathsf{sk}&=(G_1,G_{K,\Delta},I).\label{eq:scheme-secret-key}
    \end{align}
    The public parameters also specify the domain representation, the family sampler, and the associated algorithms. The secret key determines each $K_\ell$ and $\mathcal C^{(\ell)}=\operatorname{im}[\,G_1\;G_{K,\ell}\,]$, thereby defining each $\pi_\ell$. There is no need to set a separate basis for each code space. The sender and receiver share this key via an authenticated, confidential channel.

\medskip
\noindent$\mathsf{Encrypt}(\boldsymbol m,\mathsf{sk})$: For $\boldsymbol m\in\mathbb F_q^k$, independently sample $\boldsymbol r\leftarrow U(\mathbb F_q^h)$ and $\boldsymbol e\leftarrow U(E_B)$. Output
\begin{equation}
    (\boldsymbol w_{\boldsymbol m},1),
    \qquad \boldsymbol w_{\boldsymbol m}=G_1\boldsymbol m+G_{K,1}\boldsymbol r+\boldsymbol e.
    \label{eq:scheme-encrypt}
\end{equation}
Fresh encryption utilizes the masking space $K_1=K$; it employs the first $h$ columns of the storage basis. Each invocation samples fresh, independent randomness given the key.

\medskip
\noindent$\mathsf{Decrypt}((\boldsymbol w,\ell),\mathsf{sk})$: Restrict the operation to "good positions" only and solve the system of equations
\begin{equation}
    G_{1,I}\widehat{\boldsymbol m}+G_{K,\ell,I}\widehat{\boldsymbol r}=\boldsymbol w_I,
    \qquad (\widehat{\boldsymbol m},\widehat{\boldsymbol r})\in\mathbb F_q^k\times\mathbb F_q^{h_\ell},
    \label{eq:scheme-decrypt-system}
\end{equation}
where the matrix subscript $I$ indicates that the rows indexed by $I$ are retained. If there is no solution, return $\bot$; otherwise, return the message component $\widehat{\boldsymbol m}$ from any solution. Lemma~\ref{lem:prelim-recovery} guarantees the uniqueness of this component, even if the masking component is not unique.

\medskip
\noindent$\mathsf{Add}((\boldsymbol w,r),(\boldsymbol w',s);\mathsf{pp})$: Return $(\boldsymbol w+\boldsymbol w',\max\{r,s\})$.

\medskip
\noindent$\mathsf{Mult}((\boldsymbol w,r),(\boldsymbol w',s);\mathsf{pp})$: If $r+s>\Delta$, return $\bot$. Otherwise return $(\boldsymbol w\odot\boldsymbol w',r+s)$.

Both evaluation operations use only publicly available information and do not involve decryption or re-encryption. They are responsible for checking input formats and hierarchy bounds, while algebraic consistency is verified by the recipient. Constants can be provided as "fresh encryptions" of the required message vectors. The bound $L$ concerns the exposure of fresh encryptions under the same key; the algorithm itself does not track such exposure. Sums and products resulting from evaluation do not introduce new encryption randomness.

Given a key and a message, the new vector $\boldsymbol w_{\boldsymbol m}$ is uniformly distributed over the set $G_1\boldsymbol m+(K+E_B)$. Indeed, the summation map from $K\times E_B$ to $K+E_B$ is linear, and the number of pre-images corresponding to each output value is constant. Outputs obtained from valid evaluations of valid ciphertexts satisfy Equation~\eqref{eq:scheme-valid-ciphertext}, but they do not necessarily possess the distribution characteristics of fresh encryptions.

The correctness guarantee is stated in Proposition~\ref{prop:scheme-correctness} in Section~\ref{subsec:results-guarantees}; its proof is given in Section~\ref{subsec:scheme-correctness}.

\section{Security Definitions and Assumption}
\label{sec:security-model}

Before stating the security results, we specify the setup distribution, the observer's information, and the experiments used to measure distinguishing advantage. We first define the coordinate-label problem underlying our assumption, then define bounded chosen-plaintext security and the sample-distribution comparison used in the reduction. This fixes the meaning of the security claims in Section~\ref{sec:main-results}; their proofs are deferred to Section~\ref{sec:security}.

\subsection{Decisional Synchronized Codewords Problem with Independent Errors}
\label{subsec:security-assumption}

We introduce this problem to clarify the computational hardness upon which our confidentiality guarantees rely. The problem asks whether an observer can distinguish between coordinates labeled "good" and those labeled "bad" using a set of noisy codewords. The reduction in Section~\ref{subsec:security-reduction} transforms an adversary capable of distinguishing encryptions of specific messages into an algorithm that solves this "coordinate labeling problem." Consequently, the assumed hardness of this problem—based on the specific sampling distributions and auxiliary information described below—constitutes a sufficient condition for achieving bounded encryption security.

Fix a sequence of supported parameters where $q,n,k,T,(h_\ell)_{\ell=1}^{\Delta},\Delta$, and $L$ are fixed functions of $\lambda$, satisfying $k\geq1$, $1\leq T<n$, and $L\geq1$. The quantities $n,\Delta,L,\lceil\log_2q\rceil$, as well as the running time of the construction algorithm, are polynomially bounded in $\lambda$. Let $\mathcal S$ denote the joint setup distribution generated by the chosen sampler under these parameters, including the code spaces, encoding, and good positions. This distribution is supported on instances satisfying Definition~\ref{def:compatible-family}. Any constraints on $I$ required for recovery are incorporated into $\mathcal S$.

The relevant source data belongs to the zero-message space $K_1=K$. After selecting a setup, a tuple of zero-encryptions $X=(\boldsymbol x_1,\ldots,\boldsymbol x_L)$ is constructed via sampling, where:
\begin{equation}
    \boldsymbol x_j=\boldsymbol z_j+\boldsymbol e_j,
    \qquad \boldsymbol z_j\leftarrow U(K),\quad
    \boldsymbol e_j\leftarrow U(E_B),\quad j\in[L].
    \label{eq:security-zero-samples}
\end{equation}
Given the same setup, all masks and errors are mutually independent. The term "synchronized" refers to the common set $B$ of possible error positions: that is, the "good" set $I$ remains invariant across all $L$ samples. The phrase "independent errors" indicates that, for a fixed setup, fresh and independent uniformly distributed values are sampled over the set $B$ for each encryption. Since these values may be zero, the actual supports of the error vectors do not necessarily coincide. The underlying vectors $\boldsymbol z_j$ are codewords from the code $K$, whereas the noisy observations $\boldsymbol x_j$ do not necessarily belong to $K$.

In the reduction process, we utilize auxiliary information $\mathsf{aux}=(\mathsf{pp}, G_1)$. Although the actual scheme keeps $G_1$ secret, the assumption made below allows the distinguisher to access this information. This provides the simulator with the transformation mapping $\mathsf{Encode}(\boldsymbol m)=G_1\boldsymbol m$ required to answer "chosen-message queries." Throughout this process, neither the set $I$ nor the basis of the code $K$ is revealed. If the system is secure given this additional information, it implies that the original scheme (where the adversary sees only $\mathsf{pp}$) is also secure. However, the security in the presence of this additional information still needs to be demonstrated.

\begin{definition}[Decisional Synchronized Codewords Problem with Independent Errors]
\label{def:security-iid-dscp}
The decisional synchronized codewords problem with independent errors, abbreviated IID-DSCP, compares two experiments. In both experiments, a setup and the tuple $X$ from Eq.~\eqref{eq:security-zero-samples} are sampled from $\mathcal S$. Subsequently, the "good" experiment samples $i\leftarrow U(I)$, while the "bad" experiment samples $i\leftarrow U(B)$. A distinguisher $\mathcal B$ receives $(\mathsf{aux},i,X)$ and outputs a bit; its advantage is defined as
\begin{equation}
\begin{aligned}
\delta_{\mathcal S,L}(\mathcal B)
=\Bigl|&\Pr[\mathcal B(\mathsf{aux},i,X)=1\mid\mathrm{good}]\\
&-\Pr[\mathcal B(\mathsf{aux},i,X)=1\mid\mathrm{bad}]\Bigr|.
\end{aligned}
\label{eq:security-label-advantage}
\end{equation}
The probabilities above account for the randomness of the setup, the samples, the marked coordinate, and the distinguisher itself.
\end{definition}

The word \emph{decisional} refers to distinguishing these two experiments: the index $i$ is visible, but its membership in $I$ is hidden. The task concerns the label of this one coordinate and does not require recovering the entire set $I$ or decoding a codeword. The advantage measures a difference of acceptance probabilities; with equally likely experiments, the better orientation of this test has success probability $1/2+\delta_{\mathcal S,L}(\mathcal B)/2$.

\begin{assumption}[IID-DSCP with encoder information]
\label{ass:security-iid-dscp}
For a given setup distribution $\mathcal{S}$ and sample budget $L=L(\lambda)$, the advantage $\delta_{\mathcal{S},L}(\mathcal{B})$ of any probabilistic polynomial-time distinguisher $\mathcal{B}$ under Definition~\ref{def:security-iid-dscp} is negligible (where the auxiliary information $(\mathsf{pp}, G_1)$ is included).
\end{assumption}

Here, a function is said to be negligible if, for any positive constant $a$, it is eventually less than $\lambda^{-a}$. This is an assumption regarding the average-case behavior of the function family and the distribution, rather than a conclusion derived from the hardness of worst-case decoding.

The name IID-DSCP refers to the variant defined here. In Definition 5 of \cite{AAPS11}, the source coding is known, $I$ is a uniform set of size $T$, and the sampler employs distinct codewords and distinct error vectors (whose supports are exactly $B$). Our source corresponds to the mask space $K$ in that setting; independent sampling allows for repetitions and zero-error terms on $B$, while the distribution of $I$ is the one actually used by $\mathcal S$. Encoder information is also part of our experiment. Consequently, the earlier DSCP assumptions do not automatically imply Assumption \ref{ass:security-iid-dscp}; the reduction in Section~\ref{subsec:security-reduction} addresses precisely this specific variant.

\subsection{Encryption and Sample-Distinguishing Experiments}
\label{subsec:security-experiments}

The encryption experiment describes the confidentiality objective, while the sample experiment compares zero encryptions with uniform vectors. Their connection is proved later; neither definition by itself asserts hardness.

\begin{definition}[Bounded chosen-plaintext security]
\label{def:security-bounded-cpa}
Let $\Pi$ be the scheme described in Section \ref{sec:encryption}. The challenger samples a setup from $\mathcal S$ and provides the public parameters $\mathsf{pp}$ to the adversary $\mathcal A$. The adversary may request fresh encryptions of adaptively chosen messages. At some point, the adversary submits two messages $\boldsymbol m_0, \boldsymbol m_1 \in \mathbb F_q^k$; the challenger samples an independent uniform bit $\theta$ and returns a fresh ciphertext of $\boldsymbol m_\theta$. Encryption queries may continue thereafter, but the total number of fresh ciphertexts returned (including the challenge ciphertext) is at most $L$. The adversary outputs $\widehat\theta$, and its advantage is defined as
\begin{equation}
\operatorname{Adv}^{\mathrm{CPA}}_{\Pi,\mathcal A,L}
=\left|\Pr[\widehat\theta=\theta]-\frac12\right|.
\label{eq:security-cpa-advantage}
\end{equation}
The scheme is said to possess $L$-bounded IND-CPA security if this advantage is negligible for any probabilistic polynomial-time adversary. This model does not include a decryption oracle. Public addition and multiplication operations, as well as any other efficient processing of the adversary's view, may be performed locally.
\end{definition}

The convention in~\eqref{eq:security-cpa-advantage} has maximum advantage $1/2$. The reductions in Section~\ref{subsec:security-reduction} also apply to the stronger experiment that reveals $G_1$. The restriction on $L$ counts fresh encryption outputs, whereas $\Delta$ limits the correctness of prescribed evaluations. Public computations consume no fresh encryption samples, but their results remain available to the adversary.

The intermediate goal is to compare the distributions of real zero encryptions and uniform vectors. Recall that $X=(\boldsymbol x_1,\ldots,\boldsymbol x_L)$ is the tuple from~\eqref{eq:security-zero-samples}, generated under one shared setup. Let $U=(\boldsymbol u_1,\ldots,\boldsymbol u_L)$ consist of independent uniform vectors in $\mathbb F_q^n$, independent of that setup. The two views given to a test are $(\mathsf{aux},X)$ and $(\mathsf{aux},U)$, with the same marginal distribution of $\mathsf{aux}=(\mathsf{pp},G_1)$. For a binary-output algorithm $\mathcal D$, define
\begin{equation}
    \alpha_{\mathcal S,L}(\mathcal D)
      =\left|\Pr[\mathcal D(\mathsf{aux},X)=1]
                  -\Pr[\mathcal D(\mathsf{aux},U)=1]\right|.
    \label{eq:security-sample-advantage}
\end{equation}
If $\alpha_{\mathcal S, L}(\mathcal D)$ is negligible for any efficient distinguisher $\mathcal D$, we say that the two joint views are computationally indistinguishable. This requirement applies to the aggregate of all $L$ samples and their associated auxiliary information; considering individual samples in isolation is insufficient. Regarding the comparison distribution used in the proof, the "uniform tuple" does not require the code family itself to be random or unstructured.

\section{Our Results}
\label{sec:main-results}

We collect the main statements and comparisons before giving their proofs. The logical sequence is to establish which computations the scheme recovers, state the assumption sufficient for confidentiality and a restriction on that assumption, and then show that the algebraic requirements admit explicit code families. The final comparison quantifies the packing benefit over the scalar construction of~\cite{AAPS11}. The constructions establish correctness; their confidentiality remains a separate question.

\subsection{Correctness and Conditional Security}
\label{subsec:results-guarantees}

The first result applies to any compatible family and justifies the encryption and public evaluation algorithms in Section~\ref{sec:encryption}.

\begin{proposition}[Correctness of the scheme]
\label{prop:scheme-correctness}
For every instance satisfying Definition~\ref{def:compatible-family}, decrypting a fresh ciphertext yields its input message. Consider a circuit of the prescribed addition and multiplication operations on valid ciphertexts under the same key. Assign output levels recursively by taking the maximum of the input levels at an addition gate and their sum at a multiplication gate. If every resulting level is at most $\Delta$, decrypting the output yields the corresponding coordinatewise computation on the message vectors.
\end{proposition}

The proof is given in Section~\ref{subsec:scheme-correctness}. Correctness ensures that the intended message computation is recovered; the next theorem addresses whether an observer can distinguish encryptions of chosen messages. It uses the setup distribution and the stronger IID-DSCP assumption with encoder information specified in Section~\ref{sec:security-model}.

\begin{theorem}[Conditional bounded encryption security]
\label{thm:security-cpa-reduction}
For every efficient adversary $\mathcal A$ in the $L$-bounded experiment, including its stronger version revealing $G_1$, there is an efficient IID-DSCP distinguisher $\mathcal B$ such that
\begin{equation}
    \operatorname{Adv}^{\mathrm{CPA}}_{\Pi,\mathcal A,L}
       \leq 2T\,\delta_{\mathcal S,L}(\mathcal B).
    \label{eq:security-cpa-bound}
\end{equation}
Consequently, Assumption~\ref{ass:security-iid-dscp} implies $L$-bounded IND-CPA security for the scheme, including the adversary's public evaluations.
\end{theorem}

Section~\ref{subsec:security-reduction} proves the theorem by converting an encryption adversary into a sample distinguisher and then into a coordinate-label distinguisher. This is a conditional guarantee for the chosen setup distribution and fresh-sample budget, not a concrete secure parameter choice.

\subsection{A Restriction on the Security Assumption}
\label{subsec:results-restriction}

The sufficient hardness assumption must account for public products of the observed ciphertexts. The following result identifies parameter regimes in which a rank test distinguishes zero encryptions from uniform vectors without knowing the secret code spaces or error positions.

\begin{proposition}[A restriction from products of zero encryptions]
\label{prop:security-product-rank}
Put $M=n-k+1$. If
\begin{equation}
    \binom{L+\Delta}{\Delta}-1\geq M,
    \label{eq:security-monomial-count}
\end{equation}
there is an efficient sample distinguisher, using only $\mathsf{pp}$ and $X$, with
\begin{equation}
    \alpha_{\mathcal S,L}(\mathcal D)
       \geq \max\left\{0,1-\frac{M\Delta}{q}\right\}.
    \label{eq:security-product-rank-bound}
\end{equation}
Consequently, the IID-DSCP assumption cannot hold for a parameter sequence on which this lower bound is nonnegligible.
\end{proposition}

The proof and an explicit numerical example are given in Section~\ref{subsec:proof-product-rank}. For instance, when $n=10$, $k=2$, $\Delta=2$, $L=3$, and $q=101$, the stated test has sample-distinguishing advantage at least $83/101$. This contradicts the sufficient assumption in the corresponding regimes; it does not by itself give an IND-CPA attack. Parameters for which this bound is inconclusive still require an independent assessment of the assumption.

\subsection{Compatible Code Families}
\label{sec:code-families}

We next establish that the algebraic requirements are realizable. Extending the homomorphic error-correcting constructions of Yang, Gandikota, and Jones~\cite{YGJHomECC}, the next two theorems give binary Reed--Muller and Reed--Solomon-based towers. Their parameter statements are given here; the explicit polynomial spaces, message maps, and proofs appear in Sections~\ref{sec:rm-tower} and~\ref{sec:average-free-tower}.

\subsubsection{Binary Reed--Muller codes}
Binary evaluation permits message and mask terms to be separated by the supports of squarefree monomials, without an average-free exponent set. Write $\operatorname{RM}_2(s,m)$ for the binary code obtained by evaluating all multilinear polynomials of total degree at most $s$ at the points of $\mathbb F_2^m$.

\begin{theorem}[A compatible binary Reed--Muller tower]
\label{thm:rm-tower}
Let $m,r,\Delta$ be positive integers with $2\Delta r<m$, and set $n=2^m$ and $k=\binom mr$. For any good-position set $I\subseteq[n]$ satisfying
\begin{equation}
    1\leq b=n-|I|<2^{m-2\Delta r},
    \label{eq:rm-recovery-budget}
\end{equation}
there exists a compatible binary code family as in Definition~\ref{def:compatible-family}, with $\mathcal C^{(\ell)}\subseteq\operatorname{RM}_2(2\ell r,m)$ and $k$ message bits at every level. For $1\leq\ell\leq\Delta$, its dimensions and minimum distances satisfy
\begin{equation}
    \rho_\ell=\dim\mathcal C^{(\ell)}
       =\sum_{j=r}^{2\ell r}\binom mj,\qquad
    h_\ell=\rho_\ell-k,\qquad
    d_\ell=d(\mathcal C^{(\ell)})=2^{m-2\ell r}.
    \label{eq:rm-parameters}
\end{equation}
Here $h_\ell=\dim K_\ell$. The minimum distance between distinct message cosets, $\min_{\boldsymbol c\in\mathcal C^{(\ell)}\setminus K_\ell}\operatorname{wt}(\boldsymbol c)$, also equals $d_\ell$.
\end{theorem}

The construction and proof are in Section~\ref{sec:rm-tower}. Its exact distance displays the tradeoff between the number of packed message bits and the supported computation degree.

\subsubsection{Reed--Solomon subcodes}
Univariate evaluation gives a second construction, where an average-free set prevents mixed products from contributing to the intended message coordinates. The required condition is Definition~\ref{def:average-free-exponents}.

\begin{theorem}[A compatible tower from average-free exponents]
\label{thm:average-free-tower}
Let $\Delta,N$ be positive integers, put $D=2\Delta$, and choose a nonempty $D$-average-free set $S\subseteq[N]$ of size $k$. Let $q$ be a prime power and choose integers $n,T$ satisfying
\begin{equation}
    D^2N<T<n\leq q-1.
    \label{eq:tower-evaluation-budget}
\end{equation}
For any $n$ distinct nonzero evaluation points in $\mathbb F_q$ and any good-position set of size $T$, there exists a compatible family of Reed--Solomon subcodes as in Definition~\ref{def:compatible-family}, carrying $k$ message symbols at every level. Writing $jS$ for the sumset defined in the preliminaries, the code and mask dimensions are
\[
    r_\ell=\dim\mathcal C^{(\ell)}
        =\sum_{j=1}^{2\ell}|jS|,
    \qquad h_\ell=\dim K_\ell=r_\ell-k,
    \qquad 1\leq\ell\leq\Delta.
\]
Their minimum distances satisfy
\begin{equation}
    d(\mathcal C^{(\ell)})\geq n-2\ell DN,
    \qquad d(\mathcal C^{(\Delta)})\geq n-D^2N.
    \label{eq:tower-distance-bounds}
\end{equation}
The distance between distinct message cosets is at least the corresponding bound in~\eqref{eq:tower-distance-bounds}.
\end{theorem}

The construction and proof are in Section~\ref{sec:average-free-tower}. Here the distance is a lower bound, and the attainable message dimension depends on the available average-free sets. Neither theorem establishes the IID-DSCP assumption for its family; the later examples also discuss concrete attacks.

\subsection{Message Packing and Comparison}
\label{sec:packing-comparison}

We first compare the two code towers, then quantify the representation benefit relative to the scalar construction of Armknecht et al.~\cite{AAPS11}, and finally compare ciphertext and key storage under explicit counting conventions. Yang, Gandikota, and Jones~\cite[Sec.~IV, Def.~5 and Lemma~6]{YGJHomECC} describe messages through a quotient $\mathcal C_2/K$ with a complementary space $\mathcal C_1$. Its dimension is $k=\dim\mathcal C_1$, so one coset represents a vector of $k$ field symbols. Theorems~\ref{thm:rm-tower} and~\ref{thm:average-free-tower} give two extensions supporting public computation on randomly masked representatives. Their algebraic conditions and achievable parameters differ, as summarized in Table~\ref{tab:code-family-comparison}.

\begin{table}[tbp]
    \centering
    \caption{Comparison of the two sufficient code-tower constructions, extending the RM and RS coset constructions in~\cite{YGJHomECC}. Distance refers to the highest supported level.}
    \label{tab:code-family-comparison}
    \small
    \renewcommand{\arraystretch}{1.3}
    \begin{tabular*}{\textwidth}{@{\extracolsep{\fill}}lll@{}}
        \hline
        Property & Binary RM tower & RS-based tower \\
        \hline
        Field & $\mathbb F_2$ & $\mathbb F_q$, $n\leq q-1$ \\
        Evaluation domain & $\mathbb F_2^m$ & $n$ distinct nonzero field points \\
        Polynomial description & Multivariate, squarefree & Univariate, selected exponents \\
        Message dimension $k$ & $\binom mr$ & $|S|\leq a_{2\Delta}(N)$ \\
        Length & $n=2^m$ & $4\Delta^2N+2\leq n\leq q-1$ \\
        Terminal distance & $d_\Delta=2^{m-2\Delta r}$ & $d_\Delta\geq n-4\Delta^2N$ \\
        Message-product separation & Unions of supports & Average-free sums of exponents \\
        AP-free set required here & No & $2\Delta$-average-free set \\
        Ciphertext-vector bit storage & $n$ & $n\lceil\log_2q\rceil$ \\
        Confidentiality established & No & No \\
        \hline
    \end{tabular*}
\end{table}

The RM construction has a direct support-based proof and uses binary arithmetic. The RS-based construction allows a choice of evaluation length within its field and degree budgets, but the message dimension depends on the available average-free sets. Neither construction is uniformly better: the desired field, length, message dimension, and computation degree determine the relevant tradeoff. In particular, an RM message contains $k$ bits, whereas an RS-based message contains $k$ symbols of $\mathbb F_q$, with information content $k\log_2q$ bits. Comparing only the number of symbols or only $n$ does not match the bit payload or security level. The following packing counts apply to any family satisfying Definition~\ref{def:compatible-family}, at a common field and ciphertext length.

Table~\ref{tab:packing-comparison} considers one record containing $k$ field symbols. In the scalar baseline, each symbol is encrypted separately into a length-$n$ vector. A compatible packed instance carries the entire record in one length-$n$ vector over the same field. The comparison counts ciphertext field elements and field operations, excluding level labels and other metadata. It assumes that both representations support the operation being compared; it is not a comparison at matched security levels.

\begin{table}[tbp]
    \centering
    \caption{Representation and operation counts for one $k$-symbol record at the same field and ciphertext length $n$. The scalar baseline is from~\cite{AAPS11}; the packed representation uses the coset framework of~\cite{YGJHomECC} subject to our compatibility conditions.}
    \label{tab:packing-comparison}
    \small
    \renewcommand{\arraystretch}{1.2}
    \begin{tabular*}{\textwidth}{@{\extracolsep{\fill}}lcc@{}}
        \hline
        Quantity & Scalar construction & Packed construction \\
        \hline
        Message symbols per ciphertext & $1$ & $k$ \\
        Ciphertext vectors per record & $k$ & $1$ \\
        Transmitted field elements per record & $kn$ & $n$ \\
        Message symbols per transmitted field element & $1/n$ & $k/n$ \\
        Fresh encryption calls per record & $k$ & $1$ \\
        Public field operations per componentwise gate & $O(kn)$ & $O(n)$ \\
        Message recovery per record & $k$ scalar recoveries & One vector recovery \\
        \hline
    \end{tabular*}
\end{table}

The main saving is in the number of stored and transmitted ciphertext vectors. For example, if a compatible instance packs four message symbols at the same ciphertext length, a four-symbol record requires $n$ transmitted field elements instead of $4n$, a reduction of $75\%$ in the ciphertext-vector payload. Similarly, a supported componentwise addition or multiplication of two packed records uses one operation on length-$n$ vectors, whereas separate scalar encryptions require four such operations. These counts concern operations on corresponding message coordinates; they do not provide arbitrary operations between different coordinates of a packed record.

The receiver also processes one ciphertext vector for the whole record, sharing access to its good coordinates and any consistency check across the message symbols. This does not imply a factor-$k$ reduction in decryption arithmetic. For example, applying a precomputed dense linear decoder to recover $k$ symbols from $T$ good coordinates costs $O(kT)$ field operations, the same order as applying $k$ scalar linear decoders of cost $O(T)$ each. A stronger runtime claim would require a particular code family and its decoding algorithm.

Packing reduces the number of fresh encryptions used to transmit a record, but its security effect must be analyzed together with the available masking space. If $\dim\mathcal C_2$ is fixed, increasing $k$ reduces $h=\dim K=\dim\mathcal C_2-k$ and hence the number of clean random representatives available for each message. Moreover, Proposition~\ref{prop:security-product-rank} shows that the restriction from products of zero encryptions depends on $k$ as well as on the sample budget and supported degree. Thus the table establishes a representation benefit; a security comparison requires concrete choices of these parameters and an analysis of the resulting setup distribution.

\paragraph{Degree, ciphertext size, and key storage.}
Table~2 of~\cite{AAPS11} reports the supported degree $\mu$, the ciphertext-vector length $n$, the field-size quantity $\log_2q$, and ciphertext and key sizes for selected parameters. Their $\mu$ corresponds to our formal-degree bound $\Delta$. For example, two successive squarings have formal degree four, so this bound does not count multiplication gates or circuit depth. To compare storage, set $\Delta=\mu$ and write $b_q=\lceil\log_2q\rceil$ for the bits used to store one field element. Table~\ref{tab:storage-comparison} uses a common field and vector length to isolate representation costs. It does not assert that these choices yield secure instances of both schemes.

\begin{table}[tbp]
    \centering
    \caption{Parameter and direct-storage comparison at common $q,n$ and degree bound $\Delta=\mu$. Storage quantities are in bits. Ciphertext-level labels, public metadata, and cached decoders are excluded. This is not a comparison at matched security levels.}
    \label{tab:storage-comparison}
    \small
    \renewcommand{\arraystretch}{1.25}
    \begin{tabular*}{\textwidth}{@{\extracolsep{\fill}}p{0.40\textwidth}cc@{}}
        \hline
        Quantity & construction in~\cite{AAPS11} & Our construction \\
        \hline
        Supported formal degree & $\mu$ & $\Delta=\mu$ \\
        Message symbols per ciphertext & $1$ & $k$ \\
        Bits per field element & $b_q$ & $b_q$ \\
        Ciphertext-vector size & $nb_q$ & $nb_q$ \\
        Ciphertext-vector storage for a $k$-symbol record & $knb_q$ & $nb_q$ \\
        Good-position bitmap $I$ & $n$ & $n$ \\
        Full key, using the representations described below & $t(n+1)b_q+n$ & $n(k+h_\Delta)b_q+n$ \\
        \hline
    \end{tabular*}
\end{table}

The key-size row in Table~2 of~\cite{AAPS11} reports $n$ bits, counting the indicator of the good-position set $I$. Their formal key in Section~4.1 also contains the codeword support $\boldsymbol x$ and the message support $\boldsymbol y$. Appendix~B.1 omits these supports from its parameter discussion while explicitly retaining their secrecy. Thus its reported $n$-bit quantity does not count storage for the complete formal key. For their RM construction with $t$ polynomial variables, $\boldsymbol x$ consists of $n$ ordered points in $\mathbb F_q^t$ and $\boldsymbol y$ is one point in $\mathbb F_q^t$. Storing these points explicitly uses $t(n+1)b_q$ bits, giving the full direct-storage count in Table~\ref{tab:storage-comparison} after adding the $n$-bit bitmap. This count is derived from the stated key representation, rather than quoted as their Table~2 value.

Our key in~\eqref{eq:scheme-secret-key} stores $G_1\in\mathbb F_q^{n\times k}$, an adapted basis $G_{K,\Delta}\in\mathbb F_q^{n\times h_\Delta}$, and the same type of bitmap for $I$. The two matrices therefore use $n(k+h_\Delta)b_q$ bits. The adapted basis determines the smaller mask spaces through its column prefixes, so no separate matrix for each level is counted. Both full-key formulas describe explicit storage choices, not minimum key lengths; a compact description of a structured family could change the cost. Neither formula includes optional precomputed decryption matrices. A shared key can be reused within its allowed sample budget, so encrypting a $k$-symbol record as $k$ scalar ciphertexts does not multiply the key size by $k$.

The ciphertext-vector counts follow the convention of the earlier table by omitting the public degree label. Including it adds $\lceil\log_2\mu\rceil$ bits per scalar ciphertext and $\lceil\log_2\Delta\rceil$ bits per packed ciphertext. Thus a $k$-symbol record uses $k(nb_q+\lceil\log_2\mu\rceil)$ bits in the scalar representation and $nb_q+\lceil\log_2\Delta\rceil$ bits in the packed representation, under these direct encodings and excluding transmission overhead.

These comparisons establish storage costs and the payload benefit of packing; they do not establish a smaller secret key or a security advantage. In particular, we have not selected parameters giving our concrete code families the same security level as any column of Table~2 in~\cite{AAPS11}. Such a comparison requires a specified setup distribution, fresh-sample budget, supported degree, and attack analysis for each family. The discussion following their Table~2 also distinguishes its attack-based parameter choices from the larger fields needed to make the reduction's statistical-error term at most $2^{-s}$. Consequently, neither correctness alone nor matching $n,q,\mu$ supplies a matched-security comparison.

\section{Proofs}
\label{sec:proofs}

We now prove the results stated in Section~\ref{sec:main-results}. We first establish recovery and correctness, then prove the security reduction and the product-rank restriction. Finally, we give the explicit RM and RS constructions and verify their parameter and compatibility claims. The construction examples follow their respective proofs.

\subsection{Recovery and Correctness}
\label{subsec:scheme-correctness}

We begin with the recovery criterion from the preliminaries, which is the basis of the decryption argument.

\begin{proof}[Proof of Lemma~\ref{lem:prelim-recovery}]
Under the inclusion, any two compatible clean codewords differ by an element of $\mathcal C\cap E_B\subseteq\mathcal K$ and therefore have the same message. The original codeword is compatible, proving existence. Conversely, if $\boldsymbol v\in(\mathcal C\cap E_B)\setminus\mathcal K$, the identical received vector $\boldsymbol v=\boldsymbol 0+\boldsymbol v=\boldsymbol v+\boldsymbol 0$ can encode either message $\boldsymbol 0$ or $\varphi(\boldsymbol v)\ne\boldsymbol 0$, making recovery impossible.
\end{proof}

\begin{proof}[Proof of Proposition~\ref{prop:scheme-correctness}]
A fresh ciphertext satisfies~\eqref{eq:scheme-valid-ciphertext} at level $1$. More generally, any valid level-$\ell$ vector can be written as $\boldsymbol w=\boldsymbol c+\boldsymbol e$, where $\boldsymbol c\in\mathcal C^{(\ell)}$, $\pi_\ell(\boldsymbol c)=\boldsymbol m$, and $\boldsymbol e\in E_B$. Restriction to $I$ removes the error, so $\boldsymbol c$ agrees with the received good coordinates. Any other compatible codeword $\widetilde{\boldsymbol c}\in\mathcal C^{(\ell)}$ satisfies $\widetilde{\boldsymbol c}-\boldsymbol c\in\mathcal C^{(\ell)}\cap E_B\subseteq K_\ell$ by~\eqref{eq:prelim-level-recovery}. Since $K_\ell=\ker\pi_\ell$, both codewords belong to the same message coset and satisfy $\pi_\ell(\widetilde{\boldsymbol c})=\pi_\ell(\boldsymbol c)=\boldsymbol m$. Thus decryption recovers the unique message, as in Lemma~\ref{lem:prelim-recovery}.

Consider valid inputs $\boldsymbol w=\boldsymbol c+\boldsymbol e$ and $\boldsymbol w'=\boldsymbol c'+\boldsymbol e'$ at levels $r,s$. For addition, both clean codewords belong to $\mathcal C^{(t)}$ with $t=\max\{r,s\}$. Linearity and consistency of the message maps give
\[
    \pi_t(\boldsymbol c+\boldsymbol c')=\pi_r(\boldsymbol c)+\pi_s(\boldsymbol c')=\boldsymbol m+\boldsymbol m'.
\]
Since $\boldsymbol e+\boldsymbol e'\in E_B$, the returned sum is valid at level $t$.

For multiplication with $r+s\leq\Delta$, the clean product lies in $\mathcal C^{(r+s)}$ by~\eqref{eq:prelim-product-inclusion} and has message $\boldsymbol m\odot\boldsymbol m'$ by~\eqref{eq:prelim-message-product}. Equation~\eqref{eq:prelim-product-errors} shows that the remaining error lies in $E_B$. Thus, the returned product is valid at level $r+s$. Induction over the circuit gates completes the proof.
\end{proof}

\subsection{Security Proofs}
\label{sec:security}
\label{subsec:security-reduction}

We prove Theorem~\ref{thm:security-cpa-reduction} by two reductions, using the experiments of Section~\ref{sec:security-model}. The first converts a chosen-plaintext adversary into a test distinguishing zero encryptions from uniform samples. The second converts that test into an IID-DSCP distinguisher. This follows the sample-pseudorandomness approach of~\cite[Sec.~5, Thms.~2--3]{AAPS11}, with proofs for the sampling distribution and auxiliary information used here.

\subsubsection{Reduction to the coordinate-label problem}

To link this comparison to cryptographic security, we require a transformation that can both convert "zero" samples into ciphertexts of the target message and preserve uniformity when the input is uniform. The following proposition explicitly details these two properties. This transformation plays a corresponding role in Proposition 1 of \cite{AAPS11}; in the present case, however, since the structured input already represents "zero," only an additive shift is required.

\begin{proposition}[From zero samples to message encryptions]
\label{prop:security-zero-transformation}
Fix a valid setup and a message $\boldsymbol m\in\mathbb F_q^k$. Define
\begin{equation}
    \tau_{G_1}(\boldsymbol v,\boldsymbol m)
       =G_1\boldsymbol m+\boldsymbol v,
    \qquad \boldsymbol v\in\mathbb F_q^n.
    \label{eq:security-zero-transformation}
\end{equation}
This transformation has the following properties:
\begin{enumerate}
    \item If $\boldsymbol v=\boldsymbol z+\boldsymbol e$, with independent $\boldsymbol z\leftarrow U(K)$ and $\boldsymbol e\leftarrow U(E_B)$, then $(\tau_{G_1}(\boldsymbol v,\boldsymbol m),1)$ has exactly the distribution of a fresh encryption of $\boldsymbol m$.
    \item If $\boldsymbol v\leftarrow U(\mathbb F_q^n)$, then $\tau_{G_1}(\boldsymbol v,\boldsymbol m)$ is uniform on $\mathbb F_q^n$.
\end{enumerate}
\end{proposition}
\begin{proof}
For the first property, substitution gives
\[
    \tau_{G_1}(\boldsymbol z+\boldsymbol e,\boldsymbol m)
       =G_1\boldsymbol m+\boldsymbol z+\boldsymbol e.
\]
The two random terms have exactly the independent distributions required by $\mathsf{Encrypt}$, proving equality of the full encryption distributions.

For the second property, translation by $G_1\boldsymbol m$ is a bijection of $\mathbb F_q^n$, with inverse $\boldsymbol a\mapsto\boldsymbol a-G_1\boldsymbol m$. Thus, for every $\boldsymbol a\in\mathbb F_q^n$,
\[
    \Pr[\tau_{G_1}(\boldsymbol v,\boldsymbol m)=\boldsymbol a]
       =\Pr[\boldsymbol v=\boldsymbol a-G_1\boldsymbol m]
       =q^{-n}.
\]
This proves uniformity. Both arguments also apply after conditioning on a preceding interaction and the requested message, provided that $\boldsymbol v$ still has the stipulated fresh distribution under that conditioning and the fixed setup.
\end{proof}

The simulator computes $\tau_{G_1}$ using the encoder supplied in $\mathsf{aux}$. This construction uses zero samples from $K+E_B$ and therefore does not require the special full-support encoding of zero used by the earlier paper's transformation. The next lemma uses Proposition~\ref{prop:security-zero-transformation} to turn any successful encryption attack into a test distinguishing $X$ from $U$.

\begin{lemma}[From encryption attacks to sample distinguishing]
\label{lem:security-encryption-to-samples}
For every efficient adversary $\mathcal A$ in the $L$-bounded encryption experiment, including its stronger version revealing $G_1$, there is an efficient sample distinguisher $\mathcal D$ such that
\begin{equation}
    \operatorname{Adv}^{\mathrm{CPA}}_{\Pi,\mathcal A,L}
       =\alpha_{\mathcal S,L}(\mathcal D).
    \label{eq:security-encryption-to-samples}
\end{equation}
Consequently, computational indistinguishability of $(\mathsf{aux},X)$ and $(\mathsf{aux},U)$ is sufficient for $L$-bounded IND-CPA security.
\end{lemma}
\begin{proof}
Given $(\mathsf{aux}, Y)$ (where $Y$ is either $X$ or $U$), we construct a sample distinguisher $\mathcal D$. This distinguisher runs the adversary $\mathcal A$, providing it with the predetermined initial information; it samples a challenge bit $\theta$ internally and uses the next unused sample $\boldsymbol y_j$ to respond to each encryption request for a message $\boldsymbol m$ as follows:
\begin{equation}
(\tau_{G_1}(\boldsymbol y_j,\boldsymbol m),1)
=(G_1\boldsymbol m+\boldsymbol y_j,1).
\label{eq:security-simulated-encryption}
\end{equation}
The challenge phase also processes $\boldsymbol m_\theta$ using the same formula. Each response consumes one sample, requiring at most $L$ samples in total. The distinguisher outputs $1$ if and only if $\widehat\theta=\theta$.

When $Y=X$, given the secret setup, each unused sample constitutes an independent zero-encryption. According to the first property of Proposition \ref{prop:security-zero-transformation}, the response constitutes a fresh encryption of $\boldsymbol m$. This conclusion holds for adaptive requests as well: given the fixed setup, the unused randomness is independent of all previous requests and responses. Consequently, the resulting interaction follows the distribution of the real encryption experiment.

When $Y=U$, each unused sample remains uniformly distributed and independent of the prior interaction and auxiliary information. According to the second property of Proposition \ref{prop:security-zero-transformation}, the response to any requested message follows a uniform distribution. In particular, the response in the challenge phase follows the same distribution regardless of the value of $\theta$. An inductive argument regarding subsequent adaptive queries shows that the complete view is independent of $\theta$; thus, the success probability is $1/2$. The public evaluation process is based on this view and performs computations using established public algorithms; it does not generate responses that depend on secret information, and is therefore also included in the aforementioned simulation process.

In summary, when $Y=X$, $\mathcal D$ accepts (outputs $1$) with the success probability of a real adversary; when $Y=U$, it accepts with probability $1/2$. Taking the absolute difference yields
\[
\alpha_{\mathcal S,L}(\mathcal D)
=\left|\Pr[\widehat\theta=\theta]-\frac12\right|
=\operatorname{Adv}^{\mathrm{CPA}}_{\Pi,\mathcal A,L},
\]
which proves \eqref{eq:security-encryption-to-samples}. If every efficient sample distinguisher has a negligible advantage, then every efficient encryption adversary also has a negligible advantage.
\end{proof}

Next, we need to show, based on Assumption \ref{ass:security-iid-dscp}, that $(\mathsf{aux},X)$ and $(\mathsf{aux},U)$ are computationally indistinguishable. The logic is as follows: if there existed an efficient algorithm capable of distinguishing these sample distributions with non-negligible advantage, we could use it to distinguish "good" coordinate labels from "bad" ones with non-negligible advantage, thereby contradicting the aforementioned assumption. The lemma below presents this algorithmic reduction and quantifies the loss in advantage. For the independent-error experiment defined here, this lemma plays a role analogous to the sample pseudorandomness result found in Theorem 2 of~\cite{AAPS11}.

\begin{lemma}[From sample distinguishing to coordinate distinguishing]
\label{lem:security-coordinate-reduction}
For every efficient $\mathcal D$, there is an efficient IID-DSCP distinguisher $\mathcal B$ such that
\begin{equation}
    \alpha_{\mathcal S,L}(\mathcal D)
       =2T\,\delta_{\mathcal S,L}(\mathcal B).
    \label{eq:security-coordinate-reduction}
\end{equation}
The construction invokes $\mathcal D$ once and uses polynomial additional time.
\end{lemma}
\begin{proof}
For $J\subseteq[n]$, let $R_J(X)$ replace coordinate $i$ in every sample by a new independent uniform field element, for each $i\in J$. All replacement randomness is independent of the original experiment. Let $H_j$ be the distribution obtained by replacing a uniform $j$-element set of coordinates, and put $p_j=\Pr[\mathcal D(\mathsf{aux},H_j)=1]$. Thus $H_0$ is the real tuple, $H_n$ is the uniform tuple independent of $\mathsf{aux}$, and the signed endpoint gap is $a=p_0-p_n$.

Independently choose $j\leftarrow U(\{0,\ldots,n-1\})$ and $i\leftarrow U([n])$, then choose a uniform $j$-element set $J\subseteq[n]\setminus\{i\}$. For fixed $j$, $J$ and $J\cup\{i\}$ have the uniform laws required for $H_j$ and $H_{j+1}$. Indeed, for fixed sets $J_0,S$ of sizes $j,j+1$,
\[
    \Pr[J=J_0\mid j]=\frac{n-j}{n\binom{n-1}{j}}=\frac1{\binom nj},
    \qquad
    \Pr[J\cup\{i\}=S\mid j]=\frac{j+1}{n\binom{n-1}{j}}=\frac1{\binom n{j+1}}.
\]
The average signed acceptance gap before and after replacing the marked coordinate is therefore
\begin{equation}
    \frac1n\sum_{j=0}^{n-1}(p_j-p_{j+1})=\frac an.
    \label{eq:security-hybrid-gap}
\end{equation}

Write $u_{\mathrm{good}},v_{\mathrm{good}}$ for the acceptance probabilities before and after that replacement conditional on $i\in I$, and define $u_{\mathrm{bad}},v_{\mathrm{bad}}$ similarly. Set $g=u_{\mathrm{good}}-v_{\mathrm{good}}$ and $d_*=u_{\mathrm{bad}}-v_{\mathrm{bad}}$. These quantities are signed changes in acceptance probability, not success probabilities. A bad coordinate already consists of $L$ independent uniform values, jointly independent of the other coordinates conditional on the full setup: condition further on all clean masks and on the errors at other coordinates, and the independent uniform errors at the marked coordinate still make its values uniform. Replacing that coordinate preserves the joint distribution, including $\mathsf{aux}$ and $i$, so $d_*=0$.

For every setup, a uniform index belongs to $I$ with probability $T/n$. Because this probability is constant across setups, conditioning on a good or bad marked index does not change the setup distribution. The conditional views are exactly those in Definition~\ref{def:security-iid-dscp}, followed by the same replacement procedure. By total probability,~\eqref{eq:security-hybrid-gap} becomes
\[
    \frac an=\frac Tn g+\frac{n-T}{n}d_*=\frac Tn g,
    \qquad\text{hence }a=Tg.
\]

On input $(\mathsf{aux},i,X)$, the distinguisher $\mathcal B$ samples $j$ and $J$ as above and flips a fair coin. In one branch it returns $\mathcal D(\mathsf{aux},R_J(X))$; in the other it returns $1-\mathcal D(\mathsf{aux},R_{J\cup\{i\}}(X))$. Its good and bad acceptance probabilities are respectively
\[
    \frac12u_{\mathrm{good}}+\frac12(1-v_{\mathrm{good}})=\frac12+\frac g2,
    \qquad
    \frac12u_{\mathrm{bad}}+\frac12(1-v_{\mathrm{bad}})=\frac12.
\]
Thus $\delta_{\mathcal S,L}(\mathcal B)=|g|/2$, whereas $\alpha_{\mathcal S,L}(\mathcal D)=|a|=T|g|$, proving the identity. The algorithm need not know the sign or value of $g$, and the fair-coin construction makes only one call to $\mathcal D$.
\end{proof}

Under Assumption~\ref{ass:security-iid-dscp}, Lemma~\ref{lem:security-coordinate-reduction} shows that $(\mathsf{aux},X)$ and $(\mathsf{aux},U)$ are computationally indistinguishable. Indeed, if there existed an efficient sample distinguisher $\mathcal D$ with non-negligible advantage $\varepsilon(\lambda)$, the lemma would imply the existence of an efficient IID-DSCP distinguisher $\mathcal B$ satisfying
\[
\delta_{\mathcal S,L}(\mathcal B)=\frac{\varepsilon(\lambda)}{2T}.
\]
Since $T$ is polynomially bounded in $\lambda$, this advantage remains non-negligible, contradicting the assumption about the setting's distributions, auxiliary information, and sample budget. Thus, the assumed hardness rules out such sample distinguishers. While the lemma establishes a consequence of this assumption, it does not prove that IID-DSCP is hard for specific code families or parameter choices.

The precise preservation of "bad coordinates" stems from the independent uniform errors we employ. Consequently, the proof requires neither the punctured-code assumption nor the field-size correction involved in Theorem 2 of~\cite{AAPS11}.

We now combine the two reductions to prove the security theorem stated in Section~\ref{subsec:results-guarantees}.

\begin{proof}[Proof of Theorem~\ref{thm:security-cpa-reduction}]
Given $\mathcal A$, Lemma~\ref{lem:security-encryption-to-samples} constructs a sample distinguisher $\mathcal D$ with $\operatorname{Adv}^{\mathrm{CPA}}_{\Pi,\mathcal A,L}=\alpha_{\mathcal S,L}(\mathcal D)$. Lemma~\ref{lem:security-coordinate-reduction} then constructs an IID-DSCP distinguisher $\mathcal B$ with $\alpha_{\mathcal S,L}(\mathcal D)=2T\delta_{\mathcal S,L}(\mathcal B)$. Combining the two reductions proves~\eqref{eq:security-cpa-bound}. Under Assumption~\ref{ass:security-iid-dscp}, this is negligible because $T$ is polynomially bounded. The simulation includes public evaluations, and an adversary receiving only $\mathsf{pp}$ is covered by withholding $G_1$ when running it in the stronger experiment.
\end{proof}

\subsubsection{Proof of the product-rank restriction}
\label{subsec:proof-product-rank}

Our confidentiality objective is that an observer cannot efficiently distinguish encryptions of two chosen messages within the fresh-sample budget. An attack need not recover the secret key or the entire plaintext to violate this objective. Theorem~\ref{thm:security-cpa-reduction} provides a conditional guarantee: the chosen code family, setup distribution, and parameters must satisfy Assumption~\ref{ass:security-iid-dscp}. We therefore need to examine whether the algebraic structure of the ciphertexts permits an efficient test that contradicts this assumption. The intermediate question is whether the joint distribution of zero encryptions $X$ can be distinguished from the uniform tuple $U$, as defined in~\eqref{eq:security-sample-advantage}.

The secrecy of the mask space and that of the error positions have different roles. Knowing a basis of $K$ alone does not automatically give a decryption algorithm, because a ciphertext also contains an error at unknown positions. However, knowledge of the combined zero-message space $\mathcal Z=K+E_B$ would suffice to distinguish an encryption of zero from one of any nonzero message: the former belongs to $\mathcal Z$, whereas the latter does not, by the recovery condition. With a basis of $\mathcal Z$, this is a linear-algebra membership test, even if the observer cannot identify $K$ and $E_B$ separately. Thus hiding the error positions alone does not establish confidentiality.

Proposition~\ref{prop:security-product-rank} requires no recovery of these secret spaces or positions. Starting with zero encryptions, every supported nonconstant product still represents zero and lies in $\mathcal Z_\Delta=K_\Delta+E_B$. The errors are included in this space. Correctness requires $\mathcal C_1\cap\mathcal Z_\Delta=\{\boldsymbol0\}$, so its dimension is at most $n-k$. Any $n-k+1$ such vectors must therefore be linearly dependent, and an observer can detect this dependence by computing their rank without knowing a basis of $\mathcal Z_\Delta$.

The observer applies the same products to either real samples or uniform inputs and compares the resulting rank behavior. Although both cases start with only $L$ samples, public multiplication can supply many derived vectors. These vectors are functions of the inputs, rather than independent fresh samples; the proof accounts for this dependence when bounding the probability of rank deficiency under uniform inputs. Dimensions of spaces generated by coordinatewise products have also been used in attacks on code-based homomorphic encryption~\cite{GOT12,LageKurmann2026}.

The purpose of this analysis is to identify choices of $L$, $k$, $\Delta$, and $q$ for which our sufficient hardness assumption fails. A nonnegligible real-versus-uniform distinguishing advantage yields a nonnegligible IID-DSCP advantage through Lemma~\ref{lem:security-coordinate-reduction}, so Theorem~\ref{thm:security-cpa-reduction} cannot certify those choices. This particular test does not by itself distinguish encryptions of two chosen messages. Conversely, parameters for which its bound gives no useful test still require a separate assessment of the hardness assumption.

\begin{proof}[Proof of Proposition~\ref{prop:security-product-rank}]
Let $\mathcal Z_\Delta=K_\Delta+E_B$. The recovery condition gives $\mathcal C_1\cap\mathcal Z_\Delta=\{\boldsymbol0\}$: if $\boldsymbol c=\boldsymbol z+\boldsymbol e$ with $\boldsymbol c\in\mathcal C_1$ and $\boldsymbol z\in K_\Delta$, then $\boldsymbol e=\boldsymbol c-\boldsymbol z\in\mathcal C^{(\Delta)}\cap E_B\subseteq K_\Delta$, and hence $\boldsymbol c\in\mathcal C_1\cap K_\Delta=\{\boldsymbol0\}$. Therefore $\dim\mathcal Z_\Delta\leq n-k=M-1$.

There are $\binom{L+\Delta}{\Delta}-1$ distinct nonconstant formal monomials of total degree at most $\Delta$ in $L$ variables. Under~\eqref{eq:security-monomial-count}, choose any $M$ of them. Apply each monomial coordinatewise to the $L$ sample vectors and place the resulting vectors in an $n\times M$ matrix $V$. For real zero encryptions, correctness ensures that every column represents the zero message at a level at most $\Delta$, so every column belongs to $\mathcal Z_\Delta$. Thus $\operatorname{rank}V<M$. The test outputs $1$ exactly when this rank inequality holds. It uses only $M\leq n$ monomials, each of degree at most $\Delta$, and therefore runs in polynomial time.

For a uniform input tuple, consider the determinant $P$ of the first $M$ rows of $V$. It is a polynomial in the $ML$ independent uniform entries of these rows, of total degree at most $M\Delta$. It is not the zero formal polynomial: different terms of the determinant expansion assign distinct chosen monomials to at least one row, so their exponent patterns in the row variables differ and cannot cancel.

By the Schwartz--Zippel bound~\cite{Schwartz1980,Zippel1979}, a nonzero polynomial of total degree $d$ in independent uniform field variables vanishes with probability at most $d/q$. To see this, induct on the number of variables and write the polynomial as a degree-$t$ polynomial in its last variable. Its nonzero leading coefficient has total degree at most $d-t$ in the other variables, so it vanishes with probability at most $(d-t)/q$ by induction. Otherwise, the resulting nonzero univariate polynomial has at most $t$ roots, adding at most $t/q$ to the probability. The one-variable case is the same root bound. Applying this fact to $P$ gives $\Pr[P=0]\leq M\Delta/q$. Since $\operatorname{rank}V<M$ implies $P=0$, the real-versus-uniform acceptance gap is at least~\eqref{eq:security-product-rank-bound}. Lemma~\ref{lem:security-coordinate-reduction} then gives a coordinate-label advantage of at least this gap divided by $2T$, proving the last claim.
\end{proof}

\paragraph{Example and interpretation of the restriction.}
When $\Delta=2$, an observer can use each zero encryption itself, square it coordinatewise, or multiply it coordinatewise by another zero encryption. From $L$ fresh samples, this gives
\[
    \underbrace{L}_{\text{original vectors}}
    +\underbrace{L}_{\text{squares}}
    +\underbrace{\binom L2}_{\text{products of distinct inputs}}
    =\frac{L(L+3)}2
\]
columns, counted by their formal monomials. These columns need not be distinct or linearly independent. On real zero encryptions, all of them lie in $\mathcal Z_2=K_2+E_B$, the space of valid level-two zero-message vectors. Correctness prevents this space from containing a nonzero canonical message representative, so $\mathcal C_1\cap\mathcal Z_2=\{\boldsymbol0\}$ and $\dim\mathcal Z_2\leq n-k$. Any $n-k+1$ columns in this space must therefore be linearly dependent.

For a numerical illustration, take $n=10$, $k=2$, $\Delta=2$, $L=3$, and $q=101$. These values illustrate the restriction and are not proposed secure parameters. The observer receives three vectors $X=[\boldsymbol x_1\ \boldsymbol x_2\ \boldsymbol x_3]$ and forms the $10\times9$ matrix
\[
\begin{aligned}
    V&=[\,X\quad X_{\mathrm{sq}}\quad X_{\mathrm{mix}}\,],\\
    X_{\mathrm{sq}}&=[\,\boldsymbol x_1^{\odot2}\quad\boldsymbol x_2^{\odot2}\quad\boldsymbol x_3^{\odot2}\,],\\
    X_{\mathrm{mix}}&=[\,\boldsymbol x_1\odot\boldsymbol x_2\quad\boldsymbol x_1\odot\boldsymbol x_3\quad\boldsymbol x_2\odot\boldsymbol x_3\,],
\end{aligned}
\]
where $\boldsymbol x_i^{\odot2}=\boldsymbol x_i\odot\boldsymbol x_i$. For any valid setup with these parameters, real zero encryptions place all nine columns in a space of dimension at most $10-2=8$, so $\operatorname{rank}V<9$ always. If the three inputs are instead independent uniform vectors, Proposition~\ref{prop:security-product-rank} bounds the probability of the same rank deficiency by $M\Delta/q=18/101$.

The observer therefore outputs $1$ (guessing real zero encryptions) when $\operatorname{rank}V<9$, and outputs $0$ (guessing uniform inputs) otherwise. Its acceptance probabilities satisfy
\[
\begin{aligned}
    \Pr[\mathcal D(X)=1]&=1,\\
    \Pr[\mathcal D(U)=1]&\leq\frac{18}{101}.
\end{aligned}
\]
Here the public parameters are implicit in the test's input; it does not use $G_1$. The difference of these probabilities gives a sample-distinguishing advantage of at least $83/101\approx0.822$. This quantity measures distinguishing real samples from uniform samples, not distinguishing encryptions of two chosen messages.

Public multiplication is what makes this test possible with only three fresh samples: it supplies nine columns to compare with the dimension bound of eight. A test using only the original sample columns would require nine fresh samples to exceed that same bound. This comparison applies to this particular dimension test, rather than to all possible attacks. The derived columns are functions of the three inputs, not independent fresh encryptions; the determinant argument in the proof accounts for their dependence.

More generally, whenever the bound in Proposition~\ref{prop:security-product-rank} is nonnegligible along a parameter sequence, Lemma~\ref{lem:security-coordinate-reduction} converts the rank test into a nonnegligible IID-DSCP distinguisher. The assumption then fails for that sequence, so Theorem~\ref{thm:security-cpa-reduction} cannot certify its encryption security. The theorem itself remains valid: it is an implication whose premise is not satisfied. The reduction from an encryption attack to a sample test does not establish the converse, so this rank test alone is not an IND-CPA attack. If too few monomials are available, or the lower bound is zero because the field is small, this proposition is inconclusive; it does not establish security. Thus $L$, $k$, $\Delta$, and $q$ must be assessed together, including the information in publicly derived values.

\paragraph{Known and hidden code information.}
Section 5 in~\cite{AAPS11} distinguishes attacks with knowledge of the underlying code from attacks when that code is secret. In our analysis, the relevant distinction is between the actual experiment, which supplies $\mathsf{pp}$ and the requested ciphertexts, and a stronger experiment that also supplies $G_1$. Knowing $G_1$ reveals the encoding $\boldsymbol m\mapsto G_1\boldsymbol m$ and the canonical message space $\mathcal C_1$. It does not explicitly reveal the mask space $K$, the good positions $I$, or the full code family. This is partial disclosure of code information, and any additional structure inferable from $G_1$ must still be included in the security analysis.

The extra encoder information has a specific role in the proof: it lets the simulator turn an unused zero sample $\boldsymbol x$ into the requested ciphertext $G_1\boldsymbol m+\boldsymbol x$. This is why Assumption~\ref{ass:security-iid-dscp} includes $G_1$ in the auxiliary information. If the scheme is secure even with this additional information, it is also secure when $G_1$ is hidden, because an adversary can always ignore information it is given. Conversely, an attack that requires $G_1$ need not work when the encoder is hidden, and no numerical security gain from hiding it is established here.

The rank test of Proposition~\ref{prop:security-product-rank} is particularly relevant to this distinction: it uses only public parameters and the sample vectors, so hiding $G_1$ does not prevent that test. Its conclusion remains about the real-versus-uniform sample problem. To obtain a concrete encryption-security claim from Theorem~\ref{thm:security-cpa-reduction}, we still need a specified code family and setup distribution for which the IID-DSCP assumption, with its stated auxiliary information and sample budget, is justified. Recovery and homomorphic correctness alone do not establish that hardness.

\subsection{Proof of Theorem~\ref{thm:rm-tower}: the RM Construction}
\label{sec:rm-tower}

We construct the binary family promised by Theorem~\ref{thm:rm-tower}, extending the RM coset construction in~\cite[Sec.~VI-A]{YGJHomECC}. Message coefficients occupy one homogeneous degree, while masks occupy higher degrees. We first specify the code spaces and then verify their dimensions, distances, and message operations.

\subsubsection{Boolean Monomials and the Compatible Tower}
\label{subsec:rm-construction}

We use the polynomial description of binary Reed--Muller codes~\cite[Sec.~II]{AbbeShpilkaYe2021}. Fix an ordering of the $2^m$ points of $\mathbb F_2^m$, and let $\operatorname{ev}$ evaluate a polynomial at these points. A polynomial is multilinear if each variable has exponent at most one in every monomial. On Boolean evaluation points, every positive power of a coordinate equals that coordinate. Thus replacing each positive exponent by one, and collecting equal monomials over $\mathbb F_2$, preserves the evaluation vector. The resulting multilinear polynomial is unique, as proved below. This is equality of functions on $\mathbb F_2^m$; the original formal polynomial may contain higher powers. For $S\subseteq[m]$, write $x^S=\prod_{i\in S}x_i$, and define
\begin{equation}
    L_j=\operatorname{span}_{\mathbb F_2}
       \{\operatorname{ev}(x^S):S\subseteq[m],\ |S|=j\},
    \qquad 0\leq j\leq m.
    \label{eq:rm-homogeneous-spaces}
\end{equation}
Thus the binary Reed--Muller code of order $s$ is $\operatorname{RM}_2(s,m)=\bigoplus_{j=0}^{s}L_j$. The spaces below are subcodes of these Reed--Muller codes. The key multiplication rule is
\begin{equation}
    \operatorname{ev}(x^S)\odot\operatorname{ev}(x^U)
       =\operatorname{ev}(x^{S\cup U}).
    \label{eq:rm-support-union}
\end{equation}
Unlike the univariate construction considered next, this rule separates message and mask terms through unions of supports; it does not require an arithmetic-progression-free exponent set.

Use the parameters $m,r,\Delta,n,k$ of Theorem~\ref{thm:rm-tower}. For $1\leq\ell\leq\Delta$, define
\begin{equation}
    \mathcal C_1=L_r,\qquad
    \mathcal C^{(\ell)}=\bigoplus_{j=r}^{2\ell r}L_j,\qquad
    K_\ell=\bigoplus_{j=r+1}^{2\ell r}L_j.
    \label{eq:rm-tower-spaces}
\end{equation}
Equip each $\mathcal C^{(\ell)}$ with the map $\pi_\ell:\mathcal C^{(\ell)}\to\mathbb F_2^k$ that sends a codeword to the degree-$r$ coefficients of its unique multilinear polynomial representative, using the same fixed ordering of the $r$-element subsets at every level. Use the corresponding evaluated monomials as the message basis, and set $\mathcal C_2=\mathcal C^{(1)}$.

\begin{proof}[Proof of Theorem~\ref{thm:rm-tower}]
First, we justify the bases and uniqueness of the polynomial representation. For $U\subseteq[m]$, let $\boldsymbol1_U\in\mathbb F_2^m$ have ones exactly at the positions in $U$. The matrix whose row $U$ and column $S$ entry is $x^S(\boldsymbol1_U)$ has entry one exactly when $S\subseteq U$. For this invertibility argument, reorder its rows and columns by the same subset ordering, sorted by increasing cardinality. The matrix becomes lower triangular with diagonal entries one, so it is invertible. Reordering does not affect linear independence. Hence the evaluated squarefree monomials form a basis of $\mathbb F_2^{2^m}$, and every evaluation vector has a unique multilinear representative. In particular, the monomial evaluations defining the message and mask spaces are linearly independent, the sums in~\eqref{eq:rm-tower-spaces} are direct, and $\dim L_j=\binom mj$. This proves the dimension formulas and nesting of the code spaces.

We now specify the bases and construct the message maps explicitly. List the $r$-element subsets in the fixed message order as $S_1,\ldots,S_k$. The message vector $\boldsymbol m=(m_{S_1},\ldots,m_{S_k})^{\mathsf T}$ uses the standard basis of $\mathbb F_2^k$. Its corresponding ordered basis in $\mathcal C_1$ is $(\operatorname{ev}(x^{S_1}),\ldots,\operatorname{ev}(x^{S_k}))$, so the generator matrix and encoder are
\begin{equation}
\begin{aligned}
    G_1&=[\,\operatorname{ev}(x^{S_1})\ \cdots\ \operatorname{ev}(x^{S_k})\,],\\
    \mathsf{Encode}(\boldsymbol m)&=G_1\boldsymbol m
       =\operatorname{ev}\left(\sum_{j=1}^k m_{S_j}x^{S_j}\right).
\end{aligned}
    \label{eq:rm-ordered-encoder}
\end{equation}
At level $\ell$, use the evaluated monomials $\operatorname{ev}(x^U)$ with $r<|U|\leq2\ell r$ as a basis of $K_\ell$. Order them first by degree and then lexicographically within each degree. These mask bases are prefixes of the terminal basis, as required by the setup algorithm. Thus the concatenated message and mask bases specify all clean-codeword coordinates in the internal direct sum.

For $S\subseteq[m]$, let $j(S)\in[n]$ be the index of $\boldsymbol1_S$ in the fixed evaluation ordering. The set $S$ labels a monomial, whereas $j(S)$ labels a coordinate of an evaluation vector. Define the message-extraction map explicitly by
\begin{equation}
    \pi_\ell:\mathcal C^{(\ell)}\longrightarrow\mathbb F_2^k,
    \qquad
    \pi_\ell(\boldsymbol c)
       =(c_{j(S_1)},\ldots,c_{j(S_k)})^{\mathsf T}.
    \label{eq:rm-coordinate-extraction}
\end{equation}
This selects the same $k$ evaluation positions at every level; the domain of the map changes with $\ell$. For the unique multilinear representative $f$ of a codeword in $\mathcal C^{(\ell)}$, we will verify
\begin{equation}
    \bigl(\pi_\ell(\operatorname{ev}(f))\bigr)_S
       =f(\boldsymbol1_S)=[x^S]f,
    \qquad |S|=r,
    \label{eq:rm-message-map}
\end{equation}
where $[x^S]f$ is the coefficient of $x^S$. This coefficient identity depends on the chosen monomial spaces: $f$ has no terms of degree less than $r$, and every mask term has degree greater than $r$. With $\mathcal C_2=\mathcal C^{(1)}$, these spaces form the tower
\begin{equation}
    \mathcal C_1\subseteq\mathcal C_2=\mathcal C^{(1)}
       \subseteq\mathcal C^{(2)}\subseteq\cdots\subseteq\mathcal C^{(\Delta)},
    \qquad \mathcal C^{(\ell)}=\mathcal C_1\oplus K_\ell.
    \label{eq:rm-code-tower}
\end{equation}

We now verify why the explicitly defined map~\eqref{eq:rm-coordinate-extraction} returns coefficients. A codeword $\boldsymbol c\in\mathcal C^{(\ell)}$ has the unique multilinear representative
\[
    f=\sum_{r\leq|U|\leq2\ell r}a_Ux^U.
\]
At the Boolean point $\boldsymbol1_S$, a monomial $x^U$ is one exactly when $U\subseteq S$. If $|S|=r$, every nonzero contribution must therefore have $U\subseteq S$ and $|U|\geq r$, forcing $U=S$. Consequently,
\begin{equation}
    c_{j(S)}=f(\boldsymbol1_S)
       =\sum_{\substack{U\subseteq S\\r\leq|U|\leq2\ell r}}a_U
       =a_S=[x^S]f.
    \label{eq:rm-coefficient-recovery}
\end{equation}
This proves~\eqref{eq:rm-message-map}; it is not an identification of arbitrary evaluation values with polynomial coefficients. Let $J_r=(j(S_1),\ldots,j(S_k))$ be an ordered list, and let the subscript $J_r$ select rows in that listed order. The same calculation gives
\[
    [\,G_1\;G_{K,\ell}\,]_{J_r}=[\,I_k\;0\,],
    \qquad
    \pi_\ell(G_1\boldsymbol m+G_{K,\ell}\boldsymbol r)=\boldsymbol m.
\]
Here $G_{K,\ell}$ has the ordered mask basis as its columns, and $\boldsymbol r\in\mathbb F_2^{h_\ell}$ contains the mask coefficients. Thus $\pi_\ell$ is linear and surjective, its kernel is exactly $K_\ell$, and its restriction to $\mathcal C_1$ is the inverse of the chosen encoder. Since the same coordinate selection is used at every level, the message maps agree on the nested spaces. In particular, the clean encoding is systematic in the message symbols after placing the coordinates in $J_r$ first.

Next, take levels $a,b$ with $a+b\leq\Delta$, and let $f=\sum_{U\subseteq[m]}f_Ux^U$ and $g=\sum_{V\subseteq[m]}g_Vx^V$ be the multilinear representatives of codewords in $\mathcal C^{(a)}$ and $\mathcal C^{(b)}$. Define the polynomial
\begin{equation}
    p=\sum_{U,V\subseteq[m]}f_Ug_Vx^{U\cup V}.
    \label{eq:rm-reduced-product}
\end{equation}
Here $p$ is the multilinear representative of the ordinary product $fg$: a variable occurring in both factors appears squared in $fg$, but is replaced by its first power in $p$. Equal monomials are combined using addition in $\mathbb F_2$. By~\eqref{eq:rm-support-union},
\[
    \operatorname{ev}(p)=\operatorname{ev}(fg)
       =\operatorname{ev}(f)\odot\operatorname{ev}(g).
\]
For a nonzero coefficient $f_Ug_V$, the support sizes satisfy $r\leq|U|\leq2ar$ and $r\leq|V|\leq2br$. Therefore $r\leq|U\cup V|\leq2(a+b)r$, so the evaluated product belongs to $\mathcal C^{(a+b)}$. Taking linear spans gives the required product-code inclusion.

To identify its message, fix $S\subseteq[m]$ with $|S|=r$. Equation~\eqref{eq:rm-reduced-product} gives
\[
    [x^S]p=\sum_{U\cup V=S}f_Ug_V=f_Sg_S.
\]
Indeed, any nonzero summand has $|U|,|V|\geq r$, while $U,V\subseteq S$ and $|S|=r$. These conditions force $U=V=S$. Applying the message maps thus yields
\[
    \pi_{a+b}\bigl(\operatorname{ev}(f)\odot\operatorname{ev}(g)\bigr)
       =\pi_a(\operatorname{ev}(f))\odot\pi_b(\operatorname{ev}(g)).
\]
Every monomial in a mask polynomial has degree greater than $r$, so a product involving such a monomial cannot have union size $r$. This establishes the message-product identity for all masked representatives.

For completeness, a nonzero multilinear polynomial of degree at most $s\leq m$ has evaluation weight at least $2^{m-s}$. One proves this by induction on $m$. The case $m=0$ is immediate, and for $s=m$ the bound is simply one. Otherwise write $f=g+x_m h$. If $h=0$, the two slices agree and the induction hypothesis gives $\operatorname{wt}(\operatorname{ev}(f))=2\operatorname{wt}(\operatorname{ev}(g))\geq2^{m-s}$. If $h\ne0$, then $s\geq1$, $\deg h\leq s-1$, and the two slices are $g$ and $g+h$. At every point where $h$ is one, exactly one slice is one. Hence their total weight is at least $\operatorname{wt}(\operatorname{ev}(h))\geq2^{(m-1)-(s-1)}$. This completes the induction.

Applying the bound with $s=2\ell r$ gives $d_\ell\geq2^{m-2\ell r}$. To attain this weight with a nonzero message, choose disjoint sets $S,J\subseteq[m]$ with $|S|=r$ and $|J|=2\ell r-r$, and take
\[
    f=x^S\prod_{j\in J}(1+x_j).
\]
Every monomial in its expansion has degree between $r$ and $2\ell r$, and its degree-$r$ coefficient at $S$ is one. Thus its evaluation belongs to $\mathcal C^{(\ell)}\setminus K_\ell$. It is nonzero exactly when all coordinates in $S$ are one and all those in $J$ are zero, leaving $m-2\ell r$ free coordinates. Its weight is $2^{m-2\ell r}$. This proves both asserted distance equalities.

Finally,~\eqref{eq:rm-recovery-budget} gives $b<d_\Delta\leq d_\ell$. A nonzero codeword cannot be supported inside $B$, so $\mathcal C^{(\ell)}\cap E_B=\{\boldsymbol0\}\subseteq K_\ell$ at every level. Lemma~\ref{lem:prelim-recovery} supplies message recovery, completing the compatibility requirements.
\end{proof}

To illustrate the construction, take $m=5$, $r=1$, and $\Delta=2$, so a clean codeword has length $32$ and carries five message bits. Choose the two clean polynomials
\[
    f=x_1+x_2,\qquad g=x_1+x_3,
\]
with zero masks. Their linear coefficients represent the messages $\boldsymbol m=(1,1,0,0,0)^{\mathsf T}$ and $\boldsymbol m'=(1,0,1,0,0)^{\mathsf T}$. Evaluation turns each polynomial into a vector by recording its value at every point of $\mathbb F_2^5$. For this example, list the five unit points first and the remaining points in lexicographic order, obtaining $\boldsymbol t_1,\ldots,\boldsymbol t_{32}$. Thus
\[
    \operatorname{ev}(f)=\bigl(f(\boldsymbol t_1),\ldots,f(\boldsymbol t_{32})\bigr)^{\mathsf T},
    \qquad
    \operatorname{ev}(g)=\bigl(g(\boldsymbol t_1),\ldots,g(\boldsymbol t_{32})\bigr)^{\mathsf T}.
\]
Since $r=1$, the map $\pi_\ell$ reads the values at the five unit points, which equal the five linear coefficients. Under this evaluation ordering, these are the first five entries, giving
\[
    \pi_1(\operatorname{ev}(f))=(1,1,0,0,0)^{\mathsf T},
    \qquad
    \pi_1(\operatorname{ev}(g))=(1,0,1,0,0)^{\mathsf T}.
\]

Now multiply the two evaluation vectors coordinatewise. At each evaluation point, the resulting entry is the product of the corresponding values of $f$ and $g$, so it equals the value of $fg$. Using $x_1^2=x_1$ on Boolean points gives
\[
\begin{aligned}
    fg&=x_1^2+x_1x_2+x_1x_3+x_2x_3,\\
    p&=x_1+x_1x_2+x_1x_3+x_2x_3,\\
    \operatorname{ev}(f)\odot\operatorname{ev}(g)&=\operatorname{ev}(p).
\end{aligned}
\]
The product is assigned level two. Its only nonzero linear coefficient is that of $x_1$, while the quadratic terms belong to $K_2$. Applying the output message map therefore gives
\[
    \pi_2(\operatorname{ev}(p))=(1,0,0,0,0)^{\mathsf T}
       =\boldsymbol m\odot\boldsymbol m'.
\]
Thus multiplying the clean evaluation vectors preserves the intended product of the two messages, even though the resulting polynomial also contains mask terms.

\subsubsection{Parameters and a Three-Level Example}
\label{subsec:rm-parameters}

For this tower, the three main parameters satisfy
\begin{equation}
    n=2^m,\qquad k=\binom mr,\qquad
    d_\Delta=2^{m-2\Delta r}=\frac{n}{2^{2\Delta r}}.
    \label{eq:rm-tradeoff}
\end{equation}
Fixing $n$ fixes the integer $m$. For a prescribed $\Delta$, choosing $r$ then determines both the packed message dimension and the terminal distance; an arbitrary pair $(k,d_\Delta)$ need not be achievable. For example, with a target distance $d_0\geq2$, any integer $r\geq1$ satisfying $2\Delta r\leq m-\lceil\log_2 d_0\rceil$ gives $d_\Delta\geq d_0$ and $k=\binom mr$. Within this range, increasing $r$ increases the message dimension but reduces the terminal distance. Increasing $\Delta$ at fixed $m,r$ likewise reduces the distance. No choice of an AP-free set is needed in these parameter calculations.

Take $m=8$, $r=1$, and $\Delta=3$. Then one length-$256$ binary ciphertext carries $k=8$ message bits. A fresh clean polynomial has the form
\[
    f(\boldsymbol x)=\sum_{i=1}^{8}m_i x_i
       +\sum_{1\leq i<j\leq8}z_{ij}x_i x_j.
\]
The linear coefficients carry the message and the quadratic coefficients supply the random mask. When multiplying two such polynomials, Boolean reduction using $x_i^2=x_i$ preserves the desired coefficient $m_i m'_i$, while every product with at least two distinct variables remains a mask term. For example, reducing $x_1x_2\cdot x_1$ gives $x_1x_2$, whereas $x_1x_2\cdot x_3x_4=x_1x_2x_3x_4$ is already multilinear. The latter term explains why the next code space includes degrees up to four. A further multiplication by a fresh input requires degrees up to six.

\begin{table}[tbp]
    \centering
    \caption{A binary RM correctness example with $n=256$, $k=8$, and $\Delta=3$. These parameters do not establish encryption security.}
    \label{tab:rm-tower-example}
    \small
    \renewcommand{\arraystretch}{1.2}
    \begin{tabular*}{\textwidth}{@{\extracolsep{\fill}}lcccc@{}}
        \hline
        Code & Monomial degrees & $\rho_\ell$ & $h_\ell$ & Exact distance \\
        \hline
        $\mathcal C^{(1)}=\mathcal C_2$ & $1$--$2$ & $36$ & $28$ & $d_1=64$ \\
        $\mathcal C^{(2)}$ & $1$--$4$ & $162$ & $154$ & $d_2=16$ \\
        $\mathcal C^{(3)}$ & $1$--$6$ & $246$ & $238$ & $d_3=4$ \\
        \hline
    \end{tabular*}
\end{table}

All three levels recover from any $T=253$ good positions, since the remaining three positions are fewer than the terminal distance of four. The message space remains eight-dimensional throughout the tower. Its higher levels provide more room for mask terms created by products, at the cost of smaller distance.

This example also illustrates why algebraic correctness must be separated from confidentiality. At each of the eight unit evaluation points $\boldsymbol e_i$, every mask monomial vanishes and the clean value is exactly $m_i$. If the evaluation ordering is public, an adversary distinguishes a challenge encryption of $\boldsymbol0\in\mathbb F_2^8$ from one of $\boldsymbol1\in\mathbb F_2^8$ by taking the majority of these eight ciphertext entries: at most three are corrupted, so at least five equal the challenge bit. The success probability is one. The same attack applies in the stronger experiment revealing the canonical $G_1$, whose rows identify the evaluation points when $r=1$. Thus these numerical parameters do not satisfy our security assumption. The argument does not by itself establish an attack when the encoding and evaluation ordering are hidden; that case needs its own analysis.

\subsection{Proof of Theorem~\ref{thm:average-free-tower}: the RS Construction}
\label{sec:average-free-tower}

We now construct the Reed--Solomon family of Theorem~\ref{thm:average-free-tower}, extending the univariate coset construction in~\cite[Sec.~VI-B]{YGJHomECC}. Here multiplication adds integer exponents. The average-free condition separates pure products from mixed terms, while a translation separates sums containing different numbers of factors.

\subsubsection{Exponent Separation and the Nested Codes}
\label{subsec:average-free-construction}

In the product of two canonical codewords considered in~\cite{YGJHomECC}, the coefficient of $x^{2a}$ carries the product of the two message coordinates indexed by $a$. A mixed term $x^{a'+a''}$ must not occupy that same exponent. The condition $a'+a''=2a\Rightarrow a'=a''=a$ is precisely the absence of a nontrivial three-term arithmetic progression in an integer exponent set. Products of more factors require the corresponding condition for longer sums, including repeated summands. We use Definition~\ref{def:average-free-exponents}, stated in the preliminaries.

Higher average-freeness is stronger than ordinary 3-AP-freeness: for example, $\{1,2,4\}$ is 3-AP-free but $1+1+4=3\cdot2$. Excluding longer arithmetic progressions alone does not give~\eqref{eq:average-free-condition}. We must also prevent collisions between sums with different numbers of terms. Even the two-element set $\{2,3\}$ has $2+2+2=3+3$. The translation in the construction below separates those different sum lengths. Thus both conditions have a specific role in the construction.

Use the parameters $\Delta,D,N,S,k,q,n,T$ of Theorem~\ref{thm:average-free-tower}, and fix the stated evaluation points $X=(x_1,\ldots,x_n)$ and a good-position set $I$ of size $T$. Set
\[
    A=(D-1)N+S\subseteq\{(D-1)N+1,\ldots,DN\}.
\]
For a positive integer $j$, write $jA=\{a_1+\cdots+a_j:a_i\in A\}$. For $1\leq\ell\leq\Delta$, define
\begin{align}
    \Sigma_\ell&=\bigcup_{j=1}^{2\ell}jA,
    &\mathcal P_\ell&=\operatorname{span}_{\mathbb F_q}\{x^u:u\in\Sigma_\ell\},\label{eq:tower-polynomial-spaces}\\
    \mathcal C^{(\ell)}&=\operatorname{ev}_X(\mathcal P_\ell),
    &\mathcal C_1&=\operatorname{ev}_X\operatorname{span}_{\mathbb F_q}\{x^a:a\in A\}.\label{eq:tower-code-spaces}
\end{align}
Index the message coordinates by a fixed ordering of $A$. If $[x^u]f$ denotes the coefficient of $x^u$ in $f$, define
\begin{equation}
    \bigl(\pi_\ell(\operatorname{ev}_X(f))\bigr)_a
       =\sum_{j=1}^{2\ell}[x^{ja}]f,
    \qquad a\in A,\quad f\in\mathcal P_\ell,
    \label{eq:tower-message-map}
\end{equation}
and set $K_\ell=\ker\pi_\ell$. We will prove that these codes and message maps satisfy Definition~\ref{def:compatible-family}, with fresh code $\mathcal C_2=\mathcal C^{(1)}$ and canonical encoding
\[
    \mathsf{Encode}(\boldsymbol m)
       =\operatorname{ev}_X\left(\sum_{a\in A}m_a x^a\right).
\]
In particular, the tower and its decompositions are
\begin{equation}
    \mathcal C_1\subseteq\mathcal C_2=\mathcal C^{(1)}
       \subseteq\mathcal C^{(2)}\subseteq\cdots\subseteq\mathcal C^{(\Delta)},
    \qquad \mathcal C^{(\ell)}=\mathcal C_1\oplus K_\ell.
    \label{eq:average-free-code-tower}
\end{equation}

Write $r_\ell=|\Sigma_\ell|$ for the dimension claimed for the clean code; the proof will also identify it with the sumset formula in Theorem~\ref{thm:average-free-tower}.

\begin{proof}[Proof of Theorem~\ref{thm:average-free-tower}]
First, translation preserves~\eqref{eq:average-free-condition}. The sumset $jA$ is contained in the integer interval $[j((D-1)N+1),jDN]$. Consecutive such intervals are disjoint for $j<D$, since
\[
    (j+1)((D-1)N+1)-jDN=(D-1-j)N+j+1>0.
\]
Consequently, every exponent in $\bigcup_{j=1}^{D}jA$ determines its sum length $j$. Within $jA$, average-freeness implies that an exponent $ja$ can arise only from $j$ copies of $a$. Every other exponent in $jA$ is a mixed sum, and cannot represent a pure product at that length.

Every polynomial in $\mathcal P_\ell$ has degree at most $2\ell DN\leq D^2N<n$. A nonzero polynomial of this degree cannot vanish at all $n$ evaluation points, so evaluation is injective. Thus its polynomial representative and the coefficient rule~\eqref{eq:tower-message-map} are well defined. The distinct monomials indexed by $\Sigma_\ell$ give dimension $r_\ell$. Since $jA=j(D-1)N+jS$ and the sum-length intervals are disjoint, $r_\ell=|\Sigma_\ell|=\sum_{j=1}^{2\ell}|jS|$, as stated in the theorem. The canonical polynomial $\sum_{a\in A}m_a x^a$ maps to $\boldsymbol m$, proving surjectivity. Subtracting this canonical polynomial from any polynomial with message $\boldsymbol m$ leaves a kernel element. Since the canonical space maps bijectively onto $\mathbb F_q^k$, this gives the internal direct sum and $h_\ell=r_\ell-k$. The spaces are nested, and the coefficient rules agree on smaller spaces because the sum-length intervals are disjoint.

To prove multiplicative compatibility, take $r+s\leq\Delta$. A product of monomials from $\mathcal P_r$ and $\mathcal P_s$ has exponent in $jA$ for some $j\leq2(r+s)$, proving $\mathcal P_r\mathcal P_s\subseteq\mathcal P_{r+s}$. Such a product contributes to the coordinate indexed by $a$ only when all of its underlying summands equal $a$. Indeed, a mixed summand list would contradict average-freeness if its total were a pure exponent. Therefore, for $f\in\mathcal P_r$ and $g\in\mathcal P_s$,
\[
    \bigl(\pi_{r+s}(\operatorname{ev}_X(fg))\bigr)_a
       =\left(\sum_{u=1}^{2r}[x^{ua}]f\right)
        \left(\sum_{v=1}^{2s}[x^{va}]g\right).
\]
Using~\eqref{eq:prelim-evaluation-product}, this is precisely the required coordinatewise message product. The argument applies to every polynomial in the two spaces, including arbitrary kernel components.

Finally, a nonzero polynomial in $\mathcal P_\ell$ has at most $2\ell DN$ roots, which proves the distance bounds. Differences of representatives from distinct message cosets are nonzero clean codewords, so the same lower bound applies to their weights. Since $T>D^2N$, vanishing on all good coordinates forces any polynomial in the largest space to be zero. Hence $\mathcal C^{(\ell)}\cap E_B=\{\boldsymbol0\}\subseteq K_\ell$ at every level. All compatibility and recovery requirements follow.
\end{proof}

The factor $2$ in $D=2\Delta$ accounts for the fresh mask space: a fresh clean polynomial already includes products of two canonical monomials. Indeed,
\begin{equation}
    K_1=\operatorname{ev}_X\operatorname{span}_{\mathbb F_q}
       \left(\{x^{a+a'}:a\ne a',\ a,a'\in A\}
              \cup\{x^{2a}-x^a:a\in A\}\right).
    \label{eq:tower-fresh-kernel}
\end{equation}
This is the kernel form used for multiplication of two canonical codewords in~\cite{YGJHomECC}. Supporting degree-$\Delta$ computations on fresh masked inputs can involve up to $2\Delta$ canonical factors, which explains our stronger sufficient condition. In particular, one public multiplication of two fresh ciphertexts requires $\Delta\geq2$ and uses $D\geq4$ here. The theorem does not claim that this average-free order or the distance bound is optimal. The binary RM tower in Theorem~\ref{thm:rm-tower} provides a different mechanism that requires no average-free sets.

\subsubsection{Parameter Choices and a Three-Level Example}
\label{subsec:tower-parameters}

To express the tradeoff among length, message dimension, and distance, let $d_0\geq2$ be a desired lower bound on the minimum distance at the highest supported level. With $D=2\Delta$ and a fixed field large enough to supply the evaluation points, a sufficient design condition is
\begin{equation}
    k\leq a_D(N),\qquad n\geq D^2N+d_0,
    \qquad n\leq q-1.
    \label{eq:tower-parameter-region}
\end{equation}
Choose a $D$-average-free set of at least $k$ elements in $[N]$ and retain any $k$ of them; taking a subset preserves average-freeness. Theorem~\ref{thm:average-free-tower} then gives $k$ message symbols and terminal distance at least $d_0$. The integer $d_0$ is a guaranteed distance, not a formula for the exact distance of every instance. The clean-code dimension is $r_\ell$, while $k$ is the dimension of its message quotient.

\begin{table}[tbp]
    \centering
    \caption{Sufficient parameter choices for Theorem~\ref{thm:average-free-tower}, with $D=2\Delta$, $d_0\geq2$, and $n\leq q-1$. These are achievable bounds for this construction.}
    \label{tab:tower-parameter-choices}
    \small
    \renewcommand{\arraystretch}{1.25}
    \begin{tabular*}{\textwidth}{@{\extracolsep{\fill}}ll@{}}
        \hline
        Quantities fixed & Sufficient choice of the remaining parameter \\
        \hline
        Length $n$ and target $d_0$ & $N=\lfloor(n-d_0)/D^2\rfloor\geq1$, then $k\leq a_D(N)$ \\
        Length $n$ and message dimension $k$ & Choose $N$ with $a_D(N)\geq k$; obtain $d_\Delta\geq n-D^2N$ \\
        Message dimension $k$ and target $d_0$ & Choose $N$ with $a_D(N)\geq k$, then $n\geq D^2N+d_0$ \\
        \hline
    \end{tabular*}
\end{table}

In Table~\ref{tab:tower-parameter-choices}, $d_\Delta=d(\mathcal C^{(\Delta)})$, and the second row requires $n-D^2N\geq2$ to retain a nonempty error set under our sufficient recovery condition. Increasing $\Delta$ strengthens the required average-free property and increases the evaluation-degree budget in this bound. Thus, at fixed length and target distance, increasing the supported degree may reduce the achievable number of message symbols. This is the precise sense in which the choice of an exponent set controls the third parameter after the other two and the computation degree are fixed.

For an explicit example with three computation levels, take $\Delta=3$, $D=6$, $N=2$, and $S=\{1,2\}$. This two-element set is $D$-average-free: an average of its elements equals either endpoint only when all summands equal that endpoint. The translation gives $A=\{11,12\}$. Choose $q=83$, $n=80$, evaluation points $1,\ldots,80$, and $T=73$ good positions. Since $jA=\{11j,11j+1,\ldots,12j\}$, the theorem gives
\[
    \mathcal C_1\subseteq\mathcal C_2=\mathcal C^{(1)}
       \subseteq\mathcal C^{(2)}\subseteq\mathcal C^{(3)},
    \qquad k=2,
\]
with the parameters in Table~\ref{tab:tower-example}. All levels use the same length and carry the same two message symbols; the code and kernel dimensions grow while the guaranteed distance decreases.

\begin{table}[tbp]
    \centering
    \caption{A correctness example with $n=80$, $k=2$, and three computation levels.}
    \label{tab:tower-example}
    \small
    \renewcommand{\arraystretch}{1.2}
    \begin{tabular*}{\textwidth}{@{\extracolsep{\fill}}lcccc@{}}
        \hline
        Code & Maximum degree & $r_\ell$ & $h_\ell$ & Guaranteed distance \\
        \hline
        $\mathcal C^{(1)}=\mathcal C_2$ & $24$ & $5$ & $3$ & $d_1\geq56$ \\
        $\mathcal C^{(2)}$ & $48$ & $14$ & $12$ & $d_2\geq32$ \\
        $\mathcal C^{(3)}$ & $72$ & $27$ & $25$ & $d_3\geq8$ \\
        \hline
    \end{tabular*}
\end{table}

The condition $T=73>72$ permits recovery at every level while the remaining seven coordinates may contain arbitrary errors. This example illustrates correctness and the parameter counts. With its displayed evaluation points and exponents public, it is not a secure instantiation: an outsider can recover the fresh polynomial of degree at most $24$ from seven errors and then apply the message map.\footnote{Solve the linear equations $Q(x_i)=w_i E(x_i)$ with $E$ monic of degree $7$ and $\deg Q\leq31$. A solution exists by taking $E$ to vanish at the error locations and setting $Q=Ef$. For any solution, $Q-Ef$ has degree at most $31$ and vanishes at the at least $73$ uncorrupted positions, hence is zero. Polynomial division recovers $f=Q/E$ without knowing $I$.} A cryptographic instantiation requires a separate setup distribution and analysis of known-code attacks.

More generally, for the tower in Theorem~\ref{thm:average-free-tower}, revealing $G_1$ also reveals $K_1$: if $\boldsymbol g_a$ is its column indexed by $a$, then~\eqref{eq:tower-fresh-kernel} gives $K_1=\operatorname{span}\{\boldsymbol g_a\odot\boldsymbol g_{a'}:a\ne a'\}+\operatorname{span}\{\boldsymbol g_a\odot\boldsymbol g_a-\boldsymbol g_a\}$. Higher code spaces can similarly be generated from these columns. Any use of Assumption~\ref{ass:security-iid-dscp} for this family must therefore account for that exposed structure, as well as the sample and product restrictions in Section~\ref{sec:security}. The average-free condition establishes product compatibility; it does not establish the hardness assumption.

\section{Conclusion}
\label{sec:conclusion}

We propose a coset-based encryption framework in which a single ciphertext carries a vector of message symbols and supports public coordinate-wise addition and multiplication operations up to a predetermined formal degree. Compatible code spaces preserve message operations, while "good positions"—acting as secret information—enable the recovery of the resulting message cosets. Constructions based on binary Reed-Muller codes and average-free Reed-Solomon codes establish these algebraic properties and demonstrate the trade-offs between message dimension, the number of supported operations, and recovery distance. For a given field and ciphertext vector length, packing techniques reduce the number of vectors required to transmit and process a record.

Security analysis indicates that, under the established IID-DSCP assumption, the scheme satisfies conditional bounded chosen-plaintext security. However, rank tests reveal that for certain parameter choices, publicly computed product results can invalidate this assumption, even without leaking the secret code space or error location information. Consequently, the "fresh sample budget" and "supported computational operations" must be evaluated jointly. While specific code family constructions establish the scheme's correctness, concrete secure instantiations have not yet been developed. Future research should evaluate the security of these code families and sampling distributions against structural and sample-based attacks, refine the constructions where necessary, and determine whether the advantages of packing techniques persist when compared at equivalent security levels.

\bibliographystyle{IEEEtran}
\bibliography{references}

\end{document}